\documentclass[american,aps,pra,reprint,superscriptaddress]{revtex4-1}

\usepackage{amsmath,amssymb}

\usepackage{bbm}
\usepackage{color}
\usepackage{graphicx}
\usepackage[american]{babel}
\usepackage[utf8]{inputenc}
\usepackage{times}
\usepackage{epstopdf}
\usepackage{braket} 				
\usepackage{amsthm}

\newcommand{\ketbra}[2]{|#1\rangle\!\langle#2|}

\definecolor{mygrey}{gray}{0.35}
\definecolor{myblue}{rgb}{0.2,0.2,0.8}
\definecolor{myzard}{cmyk}{0,0,0.05,0}
\definecolor{mywhite}{rgb}{1,1,1}
\definecolor{myred}{rgb}{0.9,0.1,0.}
\usepackage[colorlinks=true,citecolor=myblue,linkcolor=myblue,urlcolor=myblue]{hyperref}

\newtheoremstyle{customStyle1}  
{0pt}       
{0pt}       
{\normalfont}   
{\parindent}        
{\em}  
{. --}   	 
{.5em}       
{\thmname{#1}\thmnumber{ #2}\thmnote{ (#3)}}  

\newtheorem{theorem}{Theorem}
\newtheorem{lemma}[theorem]{Lemma}
\newtheorem{definition}[theorem]{Definition}

\newtheorem{proposition}[theorem]{Proposition}

\ifx\proof\undefined
\newenvironment{proof}[1][\protect\proofname]{\par
	\normalfont\topsep6\p@\@plus6\p@\relax
	\trivlist
	\itemindent\parindent
	\item[\hskip\labelsep\scshape #1]\ignorespaces
}{%
\endtrivlist\@endpefalse
}
\providecommand{\proofname}{Proof}
\fi

\newtheorem*{theorem*}{Theorem}
\newtheorem*{lemma*}{Lemma}
\newtheorem*{corollary*}{Corollary}
\newtheorem*{observation*}{Observation}
\newtheorem*{proposition*}{Proposition}

\newcommand{\tr}{\operatorname{\bf{tr}}} 
\newcommand{\Real}{\operatorname{\bf{Re}}} 
\newcommand{\diag}{\operatorname{\bf{diag}}} 
\newcommand{\id}{\mathbbm{1}}

\newcommand{\rr}{\mathbbm{R}}

\let\oldforall\forall
\renewcommand{\forall}{\quad \oldforall}

\begin{document}
\title{Quantifying Operations with an Application to Coherence}

\author{Thomas Theurer}
\thanks{These two authors contributed equally}
\affiliation{Institute of Theoretical Physics and IQST, Universität Ulm, Albert-Einstein-Allee
	11, D-89069 Ulm, Germany}

\author{Dario Egloff}
\thanks{These two authors contributed equally}
\affiliation{Institute of Theoretical Physics and IQST, Universität Ulm, Albert-Einstein-Allee
	11, D-89069 Ulm, Germany}

\author{Lijian Zhang}
\affiliation{National Laboratory of Solid State Microstructures and College of Engineering and Applied Sciences, Nanjing University, Nanjing 210093, China }
\affiliation{Collaborative Innovation Center of Advanced Microstructures, Nanjing University, Nanjing 210093, China }

\author{Martin B. Plenio}
\affiliation{Institute of Theoretical Physics and IQST, Universität Ulm, Albert-Einstein-Allee
	11, D-89069 Ulm, Germany}

\begin{abstract}
To describe certain facets of non-classicality, it is necessary to quantify properties
of operations instead of states. This is the case if one wants to quantify how well an
operation detects non-classicality, which is a necessary prerequisite for its use
in quantum technologies. To do so rigorously, we build resource theories on the level
of operations, exploiting the concept of resource destroying maps.
We discuss the two basic ingredients of these resource theories, the free operations and
the free super-operations, which are sequential and parallel concatenations with free
operations. This leads to defining properties of functionals that are well suited to
quantify the resources of operations. We introduce these concepts at the example of
coherence. In particular, we present two measures quantifying the ability of an operation
to detect, i.e. to use, coherence, one of them with an operational interpretation, and provide methods to evaluate them. 
\end{abstract}

\maketitle
{\em Introduction.} --
In recent years, there has been an increasing interest in quantum technologies.
To investigate rigorously which properties of quantum mechanics are responsible for potential operational advantages, quantum resource theories were developed, see for example~\cite{PhysRevLett.78.2275,horodecki2003Purity,aberg2006quantifying,gour2008referenceFrames,horodecki2013quantumness,PhysRevLett.111.250404,veitch2014ResourceStabilizer,grudka2014QuantifyingContextuality,BaumgratzPhysRevLett.113.140401,del2015resourceKnowledge}.
These resource theories originate from constraints that are imposed in addition to the laws of quantum mechanics, motivated either by physical or by practical considerations.  From the constraints follow the free states and the free operations, which are the ones that can be prepared and executed without violation of the constraints.
These two main ingredients allow for the formulation of a rigorous theoretical framework in which to analyze quantitatively the amount of the resource present in quantum \emph{states} and its usefulness in operational tasks~\cite{plenio2007introduction,HorodeckiRevModPhys.81.865,StreltsovRevModPhys.89.041003}.
In addition, there exist quantum operations that can be considered resources as well, because they are not free. Therefore, a complementary question to ask is how valuable these operations are~\cite{nielsen2003quantum}. This question is often approached by the evaluation of quantities such as the resource generation capacity, i.e. the maximal increase of the resource in an input state under application of the operation, or the resource cost, i.e. the minimal amount of resources needed to simulate a non-free operation by means of free operations~\cite{PhysRevA.62.052317,PhysRevA.64.032302,1214070,PhysRevA.92.032331,xi2015entropic,garcia2016note,BU20171670,PhysRevA.95.062327}.
As we will discuss later and in the Supplemental Material (SM)~\footnote{See Supplemental Material in the appendix, which includes references~\cite{datta2017coherence,bera2018quantifying,1751-8121-50-4-045301,PhysRevA.96.022327,vidal2000entanglementMonotones,du2015coherence,du2015erratum,michelson1887relative,PhysRevLett.59.2044,hariharan2010basics,PhysRevLett.116.061102adapted}, for proofs, details, and examples},  these methods cannot be used to quantify all relevant properties of quantum operations.
Hence the situation merits a broader approach and this is why we are examining a broader framework.
More concretely, we will build formal resource theories on the level of \emph{operations}, allowing to quantify the value of operations directly.
This is also interesting from a conceptual point of view: The  goal of quantum technologies is to perform \emph{tasks} that are impossible using classical technologies. This includes sensing at high precision~\cite{RevModPhys.89.035002}, efficient processing of information, and secure transmission of data~\cite{nielsen2010quantum}. Ultimately, this is all achieved by quantum operations, i.e. dynamical resources.
Hence it seems natural to quantify the value of operations directly without the detour through states as the latter are static resources that have to be transformed into dynamic resources using free operations. Since quantum states can be seen as quantum operations with no input and a constant output (describing a quantum mechanical preparation apparatus), a resource theory on the level of operations can quantify the value of states, too, leading to a unified resource theoretic treatment of states and operations. Therefore we expect that resource theories on the level of operations will be a key method to the systematic exploration of quantum advantages.
In this work, we will exemplify the concepts and advantages of resource theories of operations at the example of coherence.

A fundamental ingredient to the departure of quantum mechanics from classical physics is the omnipresence of the superposition principle~\cite{Schrodinger1935a,PhysRevLett.116.080402}. This has led to the development of rigorous resource theories of coherence~\cite{aberg2006quantifying,BaumgratzPhysRevLett.113.140401,WinterPhysRevLett.116.120404,StreltsovRevModPhys.89.041003}, which allow to investigate the role of coherence in quantum technological applications~\cite{HilleryPhysRevA.93.012111,PhysRevLett.116.150502,Matera2058-9565-1-1-01LT01}.
These theories are formulated on the level of states and mainly focused on the inability to create coherence.
However, this is only half of the picture: To exploit coherences or more generally quantum superpositions~\cite{PhysRevLett.116.080402,PhysRevLett.119.230401} in technologies, it is both necessary to have access to operations that can create coherence and operations that can detect it in the sense that its presence  makes a difference in the measurement statistics~\cite{PhysRevX.6.041028,smirne2017coherence}. If we cannot detect or equivalently use coherence or, more generally, non-classicality, there cannot be an operational advantage in its presence. This is also reflected in ongoing  efforts to describe detectors for non-classicality~\cite{PhysRevLett.92.113602,Zhang2012,Cooper2014,Najafi2015,Izumi2018}.
As discussed in~\cite{PhysRevX.6.041028,Biswas20170170} and the SM, this is particularly clear in interferometric experiments. Therefore an answer to the question ``How well can a quantum operation detect coherence?" is needed to understand quantum advantages. 
To the best of our knowledge, and as we will discuss now, this question cannot be addressed using a resource theory on the level of states.

Although there exist mathematical frameworks for coherence theories 
on the level of states in which the free operations cannot make use of coherence~\cite{PhysRevLett.110.070502,WinterPhysRevLett.116.120404,ChitambarPhysRevLett.117.030401,PhysRevX.6.041028,PhysRevA.94.052324,PhysRevA.94.052336,PhysRevA.95.019902}, this is problematic  from a conceptual point of view:
Ideally, the presence of resources in states should be detectable by free operations, because 
this is a necessary prerequisite that such states can allow for operational advantages over free 
operations alone. If this is not possible, then it is misleading to consider a state to be resourceful
(see also~\cite{egloff2018local}), as is the situation in the theories cited above. This also 
implies that it is not possible to address the coherence detection capabilities of operations in these frameworks via the resource cost of states. 
In frameworks where coherence is useful, its detection is, as mentioned above, necessarily free, leading to a zero resource cost, which therefore cannot be used to address the coherence detection capabilities of operations either.
On the other hand, as we are interested in the question how well an operation can 
detect coherence, its coherence generation capacity cannot be the figure of merit, and therefore we cannot 
address the coherence detection capabilities of an operation based on a resource theory on the 
level of states. We refer to the SM for more details and proofs of these observations, including a 
discussion why we cannot use the coherence of the corresponding Choi state~\cite{CHOI1975285} to 
quantify the coherence of operations and why this should not be expected.

In contrast, in this Letter we will show that the coherence detection capability of operations can be quantified rigorously and that the conceptual problem discussed above vanishes if we use a resource theory on the level of operations. We will first introduce the two basic ingredients to such a theory: The free operations and the free super-operations, which map operations to operations and consist naturally of sequential and parallel concatenations with free operations~\cite{chiribella2008transforming,COECKE201659,zhuang2018resource}. From these ingredients we deduce defining properties of functionals which are well suited to quantify the value of operations. Then we present two such functionals quantifying how well an operation can detect coherence, one based on the diamond norm that can be calculated efficiently and another one, based on the induced trace norm, which has  a clear operational interpretation. We give examples for the value of operations according to these measures and conclude with an outlook on open questions.

The framework we introduce can be extended easily to operations that cannot create coherence and operations that can neither detect nor create it. We comment on results in this direction in the SM.
In a forthcoming work, our theoretical results will be used in the analysis of an experiment based 
on a photo-detector with a tunable degree of coherence detection capability~\cite{Exp}. All proofs can be found in the SM.

{\em Basic framework.} --
Since coherence is a basis dependent concept, we fix for all systems $A$ an orthonormal basis
$\ket{i}^A$ which we call incoherent. This basis is singled out by the physics of an actual system or the computational basis in a quantum algorithm.  From now on, coherences and populations will be seen with respect to the incoherent basis. The incoherent basis of a system composed of two subsystems $A$ and $B$ is given by the product basis of their incoherent bases. If it is clear from the context, we will omit the labels indicating the systems from here on. All states $\rho$ that are a statistical mixture of the incoherent basis states, i.e.
\begin{align}
	\rho=\sum_i p_i \ketbra{i}{i}
\end{align}
are called incoherent.  In the following, we make frequent use of the total dephasing operation $\Delta$
\begin{align}\label{eq:Delta}
\Delta\left(\rho\right) = \sum_i \ketbra{i}{i} \rho \ketbra{i}{i}
\end{align} which is a resource destroying map~\cite{PhysRevLett.118.060502} in coherence theory, i.e. its output is always incoherent. The total dephasing operation on a composed system is the tensor product of the total dephasing operations on the subsystems.
If we concatenate operations, we will always implicitly assume that they match, i.e. the output dimension of the first operation equals the input dimension of the second operation. In addition, we will not write the concatenation operator $\circ$ if not necessary.

To construct a resource theory that allows us to answer the question how well a quantum operation can detect coherences, we need to define the free operations and super-operations. Let us begin with the free operations.
First we notice that a Positive Operator-Valued Measure (POVM) cannot detect coherences if the measurement statistics are independent of them. This leads us to the following definition:
\begin{definition}\label{def:mincPOVM}
	A POVM given by $\{P_n\}:P_n\ge0,\sum_nP_n=\id$ is free iff
	\begin{align}
		\tr P_n \Delta \rho = \tr P_n \rho \ \forall \rho, n.
	\end{align}
\end{definition}
As one expects, all free POVMs are of the following form:
\begin{proposition} \label{prop:mincPOVM}
	A POVM is free iff
	\begin{align}
	P_n= \sum_i P_i^n\ketbra{i}{i} \ \forall n.
	\end{align}
\end{proposition}
Next we define general free operations, where we need to address subselection (by measurement results) in a consistent manner. Since the ability to do subselection depends on the actual experimental implementation, we adopt the point of view that this is a resource in itself. In general, we can have a quantum instrument $\mathcal{I}$  which allows us to do subselection according to a variable $x$, i.e we obtain with probability $p_x=\tr(\mathcal{E}_x(\rho))$ an output $\rho_x=\mathcal{E}_x(\rho)/p_x$. From the definition of the free POVMs follows that we can store the outcome $x$ in the incoherent basis of an ancillary system, which we write as
\begin{align}
	\tilde{\mathcal{I}}(\rho)=\sum_x \mathcal{E}_x(\rho)\otimes\ketbra{x}{x},
\end{align}
and implement the subselection later using a free POVM. In the special case of a POVM $\mathcal{P}$, we can represent it by
\begin{align}
	\tilde{\mathcal{P}}(\rho)=\sum_n \tr(P_n \rho) \ketbra{n}{n}.
\end{align}
Treating subselection in this way, we can reduce our analysis to trace preserving operations.

With subselection included into our framework, we call a quantum operation free if it cannot turn a free POVM into a non-free one by applying the operation prior to the measurement. This is exactly the case if it cannot transform coherences into populations \cite{PhysRevLett.110.070502}.
\begin{definition}\label{def:minc}
	A quantum operation $\Phi_\text{d-inc}$
	is called detection-incoherent iff
	\begin{align}
	\Delta \Phi_\text{d-inc} =&\Delta \Phi_\text{d-inc} \Delta.
	\end{align}
	The set of detection-incoherent operations is denoted by $\mathcal{DI}$.
\end{definition}
Note that this condition has been called \text{nonactivating} in \cite{PhysRevLett.118.060502}. With our convention for treating subselection, this includes Def.~\ref{def:mincPOVM} for POVMs.
As we mentioned in the introduction, it is both important to create and to detect coherence, therefore one can define creation-incoherent operations, i.e. operations which cannot create coherence. In coherence theory, these operations are called MIO (for maximally incoherent operations)~\cite{aberg2006quantifying,diaz2018using} or nongenerating in a general context in~\cite{PhysRevLett.118.060502}. Operations that can neither create nor detect coherence are called DIO (dephasing-covariant incoherent operations)~\cite{ChitambarPhysRevLett.117.030401,PhysRevA.94.052324,PhysRevA.94.052336,PhysRevA.95.019902}, classical operations~\cite{PhysRevLett.110.070502} or commuting~\cite{PhysRevLett.118.060502}.
\begin{definition}\label{def:creaIncInc}
	A quantum operation $\Phi_\text{c-inc}$ from system $A$ to $B$
	is called creation-incoherent, if it cannot create coherence in system $B$ if none were present in system $A$,
	\begin{align}
	\Phi_{\text{c-inc}} \Delta =&\Delta \Phi_\text{c-inc} \Delta.
	\end{align}
	A quantum operation $\Phi_\text{dc-inc}$ is called detection-creation-incoherent if it can neither detect nor create coherence,
	\begin{align}
	\Delta \Phi_\text{dc-inc} = \Phi_\text{dc-inc}  \Delta.
	\end{align}
\end{definition}
Our contribution in this work is that we show how to quantify the abilities to create and detect coherence in a rigorous manner. Note that formally, the three definitions of free operations lead to different resource theories. In the following, we will use ``free operation'' if it is unimportant which specific choice we are considering.
This allows us to introduce the second ingredient to our resource theories, the free super-operations, in a unified manner. A super-operation is free if it is a sequential and/or parallel concatenation with free operations.
\begin{definition}\label{def:freeSup}
	For free operations $\Phi$, elemental free super-operations are given by
	\begin{align}
	\mathcal{E}_{1,\Phi} \left[\Theta\right]=&\Phi \circ \Theta,&&\mathcal{E}_{2,\Phi}\left[\Theta\right]=\Theta  \circ \Phi, \nonumber \\
	\mathcal{E}_{3,\Phi}\left[\Theta\right]=&\Theta  \otimes \Phi, &&\mathcal{E}_{4,\Phi}\left[\Theta\right]=  \Phi\otimes\Theta .
	\end{align}
	A super-operation $\mathcal{F}$ is free iff it can be written as a sequence of free elemental super-operations,
	\begin{align}
	\mathcal{F}= \mathcal{E}_{i_n,\Phi_n}...\mathcal{E}_{i_3,\Phi_3} \mathcal{E}_{i_2,\Phi_2} \mathcal{E}_{i_1,\Phi_1}.
	\end{align}
\end{definition}
This definition comes from a quantum computational setting: a free super-operation is a network of free operations, into which we can plug a quantum operation.
A minimal requirement on the free super-operations is that they transform free operations into free operations, otherwise it would be possible to create resources for free. This requirement can be checked directly, see the SM. It is also straightforward to show that every free super-operation can be composed using only three elemental operations (see the SM and~\cite{chiribella2008transforming,COECKE201659}).
Whilst we focus on the ability to detect coherence in the main text, we present a few results for the other two classes of free operations in the SM (see also~\cite{PhysRevLett.110.070502,diaz2018using}).
As mentioned in the introduction, the case of coherence treated here is an example of our general setup: If one exchanges the resource destroying map
in Eq.~\ref{eq:Delta}, one can move on to Defs.~\ref{def:minc}, \ref{def:creaIncInc}, and \ref{def:freeSup}. It is also possible to define free operations without the usage of resource destroying maps and to use Def.~\ref{def:freeSup} for free super-operations~\cite{COECKE201659}.

{\em Detecting coherence.} --
To quantify the amount of a resource present in an operation, we follow the usual axiomatic approach of quantum resource theories~\cite{PhysRevLett.78.2275,horodecki2003Purity,aberg2006quantifying,gour2008referenceFrames,horodecki2013quantumness,PhysRevLett.111.250404,veitch2014ResourceStabilizer,grudka2014QuantifyingContextuality,BaumgratzPhysRevLett.113.140401,del2015resourceKnowledge}. From physical considerations, we collect a set of defining properties that every measure of the resource should obey. The first property is that the measure should be faithful, which means that it needs to be zero on the set of free operations and larger than zero on non-free operations. The second property is monotonicity under the free super-operations, i.e. the amount of resource can only decrease under the application of a free super-operation. With our convention concerning subselection, this ensures monotonicity under subselection as well~\cite{yu2016alternative}. The third property is convexity and can be seen as a matter of convenience. It ensure that mixing does not create resources. These properties lead to the following definition.
\begin{definition}\label{def:Monotone}
	A functional $M$ from quantum operations to the positive real numbers is called a resource measure
	iff
	\begin{align}
		&M\left(\Theta\right)=0 \Leftrightarrow \Theta \text{ free}, \nonumber \\
		&M\left(\Theta\right)\ge M\left(\mathcal{F}\left[\Theta\right]\right) \forall \Theta, \forall \text{ free super-operations} \ \mathcal{F} , \nonumber \\
		&M\left(\Theta\right) \text{ is convex}.
	\end{align}
\end{definition}
A functional that is a measure according to the above definition is of special interest if it has a clear operational interpretation, i.e. if the number it puts on a resource is directly connected to its value in a specific application. Often resource measures are hard to evaluate, thus measures that have a closed form expression or can be calculated efficiently using numerical methods are important as well. In the following, we will give one resource measure with respect to the ability to detect coherence that can be calculated efficiently and another one with an operational interpretation. Both involve norms on quantum operations. Therefore we review some related terminology first.
A norm $\|\cdot\|$ on quantum operations is called sub-multiplicative iff
\begin{align}
	\|\Theta_1\circ \Theta_2 \|\le \|\Theta_1\| \| \Theta_2 \| \ \forall \Theta_1,\Theta_2
\end{align}
and sub-multiplicative with respect to tensor products iff
\begin{align}
	\|\Theta_1\otimes \Theta_2 \|\le \|\Theta_1\| \| \Theta_2 \| \ \forall \Theta_1,\Theta_2.
\end{align}
Norms with the above properties can be used to define measures.
\begin{proposition}\label{prop:genMonotone}
	Let $\|\cdot\|$ denote a sub-multiplicative norm on quantum operations which is sub-multiplicative with respect to tensor products. If $\|\Phi\|\le 1$ for all $\Phi$ detection-incoherent, the functional
	\begin{align}
		M(\Theta)=\min_{\Phi \in \mathcal{DI}} \|\Delta \Theta -\Delta \Phi \|
	\end{align}
	is a measure in the detection-incoherent setting.
\end{proposition}

Choosing a particular norm in the above proposition, the so-called completely bounded trace norm or diamond norm~\cite{0036-0279-52-6-R02}, we find a measure that can be calculated efficiently.
The diamond norm is based on the trace norm, which is defined for a linear operator $A$ by~\cite{watrous2018theory}
\begin{align}\label{eq:traceNorm}
	\| A\|_1=\tr\left(\sqrt{A^\dagger A}\right).
\end{align}
The induced trace norm on a quantum operation (or more general a super-operator) $\Theta$ is, as the name suggests, defined by
\begin{align}\label{eq:traceNormOp}
	\|\Theta\|_1= \max \{ \|\Theta(X)\|_1: \| X\|_1 \le 1 \}.
\end{align}
Finally, the completely bounded trace norm or diamond norm of a quantum channel is given by
\begin{align}\nonumber
	\|\Theta^{B\leftarrow A} \|_\diamond=\sup_Z \| \Theta^{B\leftarrow A}\otimes\id^Z\|_1 = \| \Theta^{B\leftarrow A}\otimes\id^C\|_1
\end{align}
with $\dim A=\dim C$ and has multiple applications in quantum information~\cite{0036-0279-52-6-R02,PhysRevA.71.062310,watrous2018theory}. With these definitions at hand, we are ready to present our first measure.

\begin{theorem}\label{cor:DiamondMonotone}
	The functional
	\begin{align}
		M_\diamond(\Theta)=\min_{\Phi \in \mathcal{DI}} \|\Delta \Theta -\Delta \Phi \|_\diamond
	\end{align}
	is a measure in the detection-incoherent setting. We call this measure the diamond-measure.
\end{theorem}

Rather surprisingly, we show in the SM that this measure can be calculated efficiently using a semidefinite program~\cite{boyd2004convex} which is based on \cite{v005a011,CHOI1975285,PhysRevA.87.022310}.
A related measure is given in the following theorem.
\begin{theorem} \label{prop:BiasMon}
	The functional	
	\begin{align}
		\tilde{M}_\diamond(\Theta)=\min_{\Phi\in \mathcal{DI}}\|\Delta \Theta- \Delta\Phi \|_1
	\end{align}
	is a measure in the detection-incoherent setting.
	We call it the nSID-measure (non-stochasticity in detection).
\end{theorem}
As we prove in the SM, this measure has an operational interpretation in our framework.
Assume you obtain a \emph{single} copy of a quantum channel which is equal to $\Theta_0$ or $\Theta_1$ with probability $1/2$ each. The optimal probability $P_c(1/2,\Theta_0,\Theta_1)$ to correctly guess $i=0,1$ if one can perform only detection-incoherent measurements is given by (see also~\cite{Matthews2009})
\begin{align}
	P_c(1/2,\Theta_0,\Theta_1)=& \frac{1}{2}+ \frac{1}{4} \max_{\ket{\psi}} \|\Delta\left(\Theta_0-\Theta_1\right) \ketbra{\psi}{\psi} \|_1. \nonumber
\end{align}
Therefore, in a single shot regime,
$1/2+1/4\ \tilde{M}_\diamond(\Theta)$
is the optimal probability to guess correctly if one obtained $\Theta$ or the least distinguishable free operation, provided we can use only free measurements. Accordingly, the measure $\tilde{M}_\diamond$ quantifies how well the visible part of an operation can be approximated by a free one. This operational interpretation is the reason for the choice of the name nSID-measure.
Note that a similar interpretation holds for the diamond-measure with the only difference that, on the auxiliary system, non-free measurements  are allowed as well. Therefore the diamond-measure is an upper bound on the nSID-measure.
The operational interpretation of this measure (which satisfies faithfulness) proves that we can distinguish at no cost all operations that can detect coherence from those that cannot. As we argue in the introduction and the SM, this is an important property which cannot be achieved using any coherence theory on the level of states. In the SM, we give details of how this measure can be evaluated and some examples.

Now that we have described a measure with an operational interpretation, a natural question is which quantum operations maximize this measure. The answer is given by the following proposition.
\begin{proposition}\label{prop:MaxMeas}
	The maximum value of $\tilde{M}_\diamond(\Theta)$ for $\Theta$ a quantum channel with input of dimension $n$ and output of dimension $m$ is given by
	\begin{align}
	\frac{2 (N_0-1)}{N_0},
	\end{align}
	where $N_0=\min\{n,m\}$. It is both saturated by a Fourier transform in a subspace of dimension $N_0$ and by a measurement in the Fourier basis, encoding the outcomes in the incoherent basis.
\end{proposition}
For transformations on qubits, this means that the Hadamard gate is best suited to detect coherence in the sense of the nSID-measure. This can be seen as a reason why for example the Deutsch–Jozsa algorithm~\cite{deutsch1985quantum,deutsch1992rapid} not only starts but also finishes with Hadamard gates. It is not enough to create coherence, it also has to be detected, i.e. used, in order to exploit it.

{\em Conclusions.} --
We argued why the formulation of resource theories on the level of operations are a
valuable unifying concept and demonstrated at the example of coherence theory how to construct them rigorously using resource destroying maps~\cite{PhysRevLett.118.060502}. These theories are based on two main ingredients, the free operations and the free super-operations. The free super-operations are sequential and parallel concatenations with free operations, i.e. the embedding into a network of free operations. Based on physical considerations, we defined properties that a measure of resource in an operation should obey, for example monotonicity under the free super-operations.
We focused particularly on the question how well a quantum operation can detect coherence. This is important, since both the ability to create and to detect coherence are necessary prerequisites for operational advantages of quantum computation over classical computation and the latter cannot, as we have shown, be addressed using resource theories on the level of states.
We presented two measures quantifying the ability of an operation to detect coherence. The first can be calculated efficiently using a semidefinite program. The second, named the nSID-measure, can be evaluated in an iterative manner and has a clear operational interpretation. Its value determines how well we can distinguish the given quantum operation from the free operations in a single try.
Finally, we proved that Fourier transforms and measurements in a Fourier basis maximize the nSID-measure and can therefore be considered optimal in the task of measuring coherence.

Completion of the resource theories provided here is a sizable task. It includes the question of manipulation, quantification, and exploitation of the resourceful operations using free super-operations. A thorough answer to these questions may lead to a better understanding of operational advantages provided by quantum devices, which in turn may lead to improved designs. Working out our approach in scenarios different from coherence theory will shed new light on other quantum properties.

\begin{acknowledgments}
	We acknowledge helpful discussions with Andrea Smirne, Joachim Rosskopf, Feixiang Xu, and Huichao Xu. TT, DE, and MBP acknowledge financial support by the ERC Synergy Grant BioQ (grant no 319130). LZ is grateful to the financial support from National Natural Science Foundation of China under Grants No. 11690032 and No. 11474159.
\end{acknowledgments}

%

\newpage
\clearpage

\onecolumngrid
\section*{Supplemental Material: Quantifying Operations with an Application to Coherence}

In this Supplemental Material, we give the proofs of the results in the main text and some further details and examples.

\section{Simplifying results}
Here we present some results that will simplify the proofs of the results in the main text. For completeness, we first proof the following lemma, which is basically clear by definition.
\setcounter{theorem}{10}
\begin{lemma}\label{lem:incComp}
	If both $\Phi$ and $\Theta$ are free, $\Phi\circ\Theta$ and $\Phi \otimes \Theta$ are free as well.
\end{lemma}
\begin{proof}
	\begin{align}
		&\Phi_\text{c-inc} \Theta_\text{c-inc} \Delta = \Phi_\text{c-inc} \Delta \Theta_\text{c-inc} \Delta = \Delta \Phi_\text{c-inc}\Delta \Theta_\text{c-inc} \Delta = \Delta \Phi_\text{c-inc} \Theta_\text{c-inc} \Delta,\nonumber \\
		&\Delta \Phi_\text{ d-inc} \Theta_\text{ d-inc}  =  \Delta\Phi_\text{ d-inc} \Delta \Theta_\text{ d-inc}  = \Delta \Phi_\text{ d-inc} \Delta \Theta_\text{ d-inc} \Delta = \Delta \Phi_\text{ d-inc} \Theta_\text{ d-inc} \Delta, \nonumber \\
		&\left( \Phi_\text{c-inc} \otimes \Theta_\text{c-inc} \right)\Delta = \left( \Phi_\text{c-inc} \Delta \right)  \otimes \left(\Theta_\text{c-inc} \Delta\right)= \left(\Delta \Phi_\text{c-inc} \Delta \right)  \otimes \left(\Delta \Theta_\text{c-inc} \Delta\right) =\Delta \left( \Phi_\text{c-inc} \otimes \Theta_\text{c-inc} \right)\Delta,\nonumber \\
		&\Delta\left( \Phi_\text{d-inc} \otimes \Theta_\text{d-inc} \right) = \left(\Delta \Phi_\text{d-inc}  \right)  \otimes \left( \Delta \Theta_\text{d-inc} \right)= \left(\Delta \Phi_\text{d-inc} \Delta \right)  \otimes \left(\Delta \Theta_\text{d-inc} \Delta\right) =\Delta \left( \Phi_\text{d-inc} \otimes \Theta_\text{d-inc} \right)\Delta.
	\end{align}
\end{proof}

This shows that our free super-operations indeed preserve the set of free operations. In addition, it allows us to give the following simplified characterization of the free super-operations (see also \cite{chiribella2008transforming} and lemma 3.11 in~\cite{COECKE201659}).
\begin{proposition} \label{prop:FormSuperoperator}
	A super-operation  $\mathcal{F}$ is
	free if and only if it can be written as
	\begin{align}
		\mathcal{F}\left[\Theta\right]=\mathcal{F}_{\Phi_2,\Phi_1}\left[\Theta\right]:=&\Phi_2 \circ \left(\Theta\otimes\id\right)\circ\Phi_1
	\end{align}
	where $\Phi_1$ and $\Phi_2$ are free.
\end{proposition}
\begin{proof}
	First we note that
	\begin{align}
		\mathcal{F}_{\Phi_2,\Phi_1}\left[\Theta\right]=\mathcal{E}_{1,\Phi_2}\circ \mathcal{E}_{2,\Phi_1} \circ \mathcal{E}_{3,\id} \left[\Theta\right]
	\end{align}
	is always free by definition. Defining the (in all three frameworks) free operation
	\begin{align}
		\Phi_S(\rho\otimes \sigma)= \sigma \otimes \rho,
	\end{align}
	we can write all elemental free operations in this form,
	\begin{align}
		\mathcal{E}_{1,\Phi}\left[\Theta\right]= &\Phi \circ \Theta= \Phi \circ \left(\Theta \otimes 1\right)\circ \id = \mathcal{F}_{\Phi,\id}\left[\Theta\right], \nonumber \\
		\mathcal{E}_{2,\Phi}\left[\Theta\right]= & \Theta \circ \Phi= \id \circ \left(\Theta \otimes 1\right)\circ \Phi = \mathcal{F}_{\id,\Phi}\left[\Theta\right], \nonumber \\
		\mathcal{E}_{3,\Phi}\left[\Theta\right]= &\Theta \otimes \Phi= \left(\Theta \otimes \id\right)\circ \left(\id \otimes\Phi\right)=\id \circ \left(\Theta \otimes \id\right)\circ \left(\id \otimes\Phi\right) = \mathcal{F}_{\id,\id\otimes\Phi}\left[\Theta\right], \nonumber \\
		\mathcal{E}_{4,\Phi}\left[\Theta\right]= & \Phi \otimes \Theta = \Phi_S\circ \left(\Theta\otimes \Phi \right) \circ  \Phi_S = \Phi_S  \circ \left(\Theta \otimes \id\right)\circ \left(\id \otimes\Phi\right) \circ \Phi_S = \mathcal{F}_{\Phi_S,(\id\otimes\Phi)\circ \Phi_S}\left[\Theta\right].
	\end{align}
	Therefore
	\begin{align}
		\mathcal{F}_{\Phi_4,\Phi_3} \circ \mathcal{F}_{\Phi_2,\Phi_1}\left[\Theta\right]=& \Phi_4 \circ \left[ \left(\Phi_2\circ \left(\Theta \otimes \id \right)\circ\Phi_1 \right)\otimes \id \right]\circ \Phi_3 \nonumber \\
		=& \Phi_4 \circ \left[ \left(\Phi_2\otimes \id \right)\circ \left(\left(\left(\Theta \otimes \id \right)\circ\Phi_1 \right)\otimes \id \right) \right]\circ \Phi_3 \nonumber \\
		=& \Phi_4 \circ \left(\Phi_2\otimes \id \right)\circ\left[  \left(\left(\Theta \otimes \id \right)\otimes\id \right)\circ\left(\Phi_1 \otimes \id \right) \right]\circ \Phi_3 \nonumber \\
		=& \Phi_4 \circ \left(\Phi_2\otimes \id \right)\circ\left[  \Theta \otimes \id \otimes\id  \right] \circ\left(\Phi_1 \otimes \id \right) \circ \Phi_3 \nonumber \\
		=& \Phi_4 \circ \left(\Phi_2\otimes \id \right)\circ\left[  \Theta \otimes \id   \right] \circ\left(\Phi_1 \otimes \id \right) \circ \Phi_3 \nonumber \\
		=&\mathcal{F}_{\Phi_4\circ \left(\Phi_2\otimes\id\right),\left(\Phi_1\otimes\id\right)\circ\Phi_3}\left[\Theta\right]
	\end{align}
	finishes the proof.
\end{proof}
As shown in the following proposition, this simplifies the defining properties of a measure.
\begin{proposition}\label{prop:Monotone}
	$M$ is a measure of resource of operations iff
	\begin{alignat}{2}
		&M\left(\Theta\right)=0 \Leftrightarrow \Theta \text{ free},&& \nonumber \\
		&M\left(\Theta\right)\ge M\left(\mathcal{E}_{1,\Phi}\left[\Theta\right]\right) &&\forall \Theta, \forall\Phi \text{ free},\nonumber \\
		&M\left(\Theta\right)\ge M\left(\mathcal{E}_{2,\Phi}\left[\Theta\right]\right) &&\forall \Theta, \forall\Phi \text{ free},\nonumber \\
		&M\left(\Theta\right)\ge
		M\left(\mathcal{E}_{3,\id}\left[\Theta\right]\right) &&\forall \Theta, \nonumber \\
		&M\left(\Theta\right) \text{ is convex}.
	\end{alignat}
\end{proposition}
\begin{proof}
	Assume $M$ is a measure. Then, by definitions, the conditions hold. Now assume the conditions hold. Using Prop.~\ref{prop:FormSuperoperator}, we can write
	\begin{align}
		M\left(\mathcal{F}\left[\Theta\right]\right)=M\left(\mathcal{E}_{1,\Phi_2}\circ \mathcal{E}_{2,\Phi_1} \circ \mathcal{E}_{3,\id} \left[\Theta\right]\right) \le M\left(\mathcal{E}_{2,\Phi_1} \circ \mathcal{E}_{3,\id} \left[\Theta\right]\right) \le M\left(\mathcal{E}_{3,\id} \left[\Theta\right]\right) \le M\left( \Theta\right).
	\end{align}
\end{proof}

Next we have the following lemma, which we will use frequently.
\begin{lemma} \label{lem:evecM}
	The eigenvectors of $\id^A \otimes \Delta^B (\rho^{A,B})$ are separable and of the form $\ket{\phi_{a|b}}^A \otimes \ket{b}^B$.
\end{lemma}

\begin{proof}
	Define
	\begin{align}
		\rho^{A,B} =\sum_{i,j,k,l} \rho_{ijkl} \ketbra{i,j}{k,l}.
	\end{align}
	If we do a projective measurement $\{\ket{i}\bra{i}\}$ on system $B$ with outcome $b$, the post-measurement state of system $A$ is given by
	\begin{align}
		\rho_{|b}=\sum_{i,k} \rho_{ibkb}/p_b \ketbra{i}{k}=\sum_a q_{a|b} \ketbra{\phi_{a|b}}{\phi_{a|b}}
	\end{align}
	where the right side is its eigendecomposition. We define the orthonormal set of states
	\begin{align}
		\ket{\psi_{a,b}}=\ket{\phi_{a|b}}^A\otimes\ket{b}^B.
	\end{align}
	Then
	\begin{align}
		\id^A \otimes \Delta^B (\rho^{A,B}) \ket{\psi_{a,b}}=& \sum_{i,j,k} \rho_{ijkj}\ket{i,j}\braket{k,j|\phi_{a|b},b} \nonumber\\
		=& \sum_{i,j} q_{i|j} p_j \ket{\phi_{i|j},j}\braket{\phi_{i|j},j|\phi_{a|b},b} \nonumber \\
		=&q_{a|b} p_b\ket{\psi_{a,b}}.
	\end{align}
	Thus all eigenvectors of $\id^A \otimes \Delta^B (\rho^{A,B})$ are of the form $\ket{\phi_{a|b}}^A\otimes\ket{b}^B$.
\end{proof}

In general, we can check directly if a quantum operation is free: Every quantum operation $\Phi$ is linear and thus completely determined by the coefficients $\Phi_{k,l}^{i,j}$ defined through
\begin{align}
	\Phi\left(\ket{i}\bra{j}\right)=\sum_{k,l} \Phi_{k,l}^{i,j}\ketbra{k}{l}.
\end{align}
\begin{proposition}\label{prop:caracFreeOp}
	Let us represent a linear map $\Phi(\rho)$
	by the coefficients $\Phi_{a,c}^{b,d}$ as above.
	Then $\Phi$ is completely positive iff
	\begin{align}
		\Phi_{a,c}^{b,d}= \sum_n K_{n_{a,b}} K_{n_{c,d}}^*.
	\end{align}
	Under this condition,
	
	$\Phi$ is a detection-incoherent quantum operation iff
	\begin{align} \label{m-inc cara}
		\Phi_{a,a}^{b,d}=p\left(a|b\right)\delta_{b,d} \forall a,b,d.
	\end{align}
	
	$\Phi$ is a creation-incoherent quantum operation iff
	\begin{align}
		\Phi_{b,c}^{a,a}=p\left(b|a\right)\delta_{b,c} \wedge \sum_a  \Phi_{a,a}^{b,d}=\delta_{b,d}. \forall a,b,d.
	\end{align}
	
	$\Phi$ is a detection-creation-incoherent operation iff
	\begin{align}
		\Phi_{a,a}^{b,d}=p\left(a|b\right)\delta_{b,d} \forall a,b,d, \nonumber \\
		\Phi_{b,c}^{a,a}=p\left(b|a\right)\delta_{b,c} \forall a,b,d.
	\end{align}
\end{proposition}
\begin{proof}
	$\Phi$ being completely positive is equivalent to $\Phi\left(\rho\right)=\sum_n K_n \rho K_n^\dagger$ which we can write as
	\begin{align}
		\Phi(\rho)=\sum_{n,i,j,a,b,c,d} K_{n_{a,b}}  \ketbra{a}{b}\rho_{i,j}\ketbra{i}{j}K_{n_{c,d}}^* \ketbra{d}{c}= \sum_{n,a,b,c,d} K_{n_{a,b}} K_{n_{c,d}}^*   \rho_{b,d}\ketbra{a}{c}
	\end{align}
	from which follows
	\begin{align}
		\Phi_{a,c}^{b,d}= \sum_n K_{n_{a,b}} K_{n_{c,d}}^*.
	\end{align}
	Now we come to the first part of the proposition. We have
	\begin{align}
		\Delta \Phi \rho = &\sum_{a,b,d} \Phi_{a,a}^{b,d}   \rho_{b,d}\ketbra{a}{a}, \nonumber \\
		\Delta \Phi \Delta \rho =& \sum_{a,b} \Phi_{a,a}^{b,b}   \rho_{b,b}\ketbra{a}{a}.
	\end{align}
	Thus $\Phi$ being detection-incoherent is equivalent to
	\begin{align}
		\sum_{b,d} \Phi_{s,s}^{b,d}   \rho_{b,d} = \sum_{b} \Phi_{s,s}^{b,b}     \rho_{b,b} \forall s, \rho,
	\end{align}	
	which is exactly the case if this condition holds for all pure $\rho=\ketbra{\psi}{\psi}$ with
	\begin{align}
		\ket{\psi}=\sum_j \sqrt{q_j} e^{i\gamma_j}\ket{j},
	\end{align}
	meaning $\Phi$ is detection-incoherent iff
	\begin{align}\label{m-inc equiv}
		\sum_{b\ne d} \Phi_{s,s}^{b,d}  \sqrt{q_b} \sqrt{q_d} e^{i(\gamma_b-\gamma_d)} = 0\forall s, \{q_i\} \text{prob. distr.}, \gamma_i\in \rr.
	\end{align}
	From this follows the necessary condition
	\begin{align}
		0=& \frac{1}{2}\left( \Phi_{s,s}^{b,d}     e^{i(\gamma_b-\gamma_d)} + \Phi_{s,s}^{d,b}     e^{i(\gamma_d-\gamma_b)}  \right)\nonumber \\
		=&\Real \left(\sum_{n} K_{n_{s,b}} K_{n_{s,d}}^*   e^{i(\gamma_b-\gamma_d)} \right) \forall s,b\ne d, \gamma_b \in  \rr , \gamma_d \in \rr,
	\end{align}
	where $\Real$ denotes the real part, 
	and thus
	\begin{align}
		\sum_{n} K_{n_{s,b}} K_{n_{s,d}}^* = \Phi_{s,s}^{b,d}=0 \forall s,b\ne d.
	\end{align}
	In addition, since we work with trace preserving operations, we have
	\begin{align}
		p(a|b)=\bra{a}\sum_n K_n \ketbra{b}{b}K_n^\dagger\ket{a} = \sum_n K_{n_{a,b}} K_{n_{a,b}}^*= \Phi_{a,a}^{b,b}
	\end{align}
	and therefore necessarily
	\begin{align}
		\sum_{n} K_{n_{s,b}} K_{n_{s,d}}^*= \Phi_{s,s}^{b,d}  =p\left(a|b\right)\delta_{b,d} \forall a,b,d.
	\end{align}
	This is already sufficient for trace preservation, since
	\begin{align}
		\id=\sum_n K_n^\dagger K_n \Leftrightarrow \sum_{n,a} K_{n_{a,b}} K^*_{n_{a,d}}=\sum_a  \Phi_{a,a}^{b,d}=\delta_{b,d}.
	\end{align}
	To show that condition~(\ref{m-inc cara}) is also sufficient, we can plug it into condition~(\ref{m-inc equiv})
	\begin{align}
		\sum_{n,b\ne d} K_{n_{s,b}} K_{n_{s,d}}^*   \sqrt{q_b} \sqrt{q_d} e^{i(\gamma_b-\gamma_d)} = \sum_{b\ne d} p\left(s|b\right)\delta_{b,d} \sqrt{q_b} \sqrt{q_d} e^{i(\gamma_b-\gamma_d)} = 0 \forall s, \{q_i\} \text{prob. distr.}, \gamma_i\in \rr,
	\end{align}
	and see that it is satisfied.
	
	Now we come to the second part.
	Assume $\Phi$ is creation-incoherent. From
	\begin{align}
		\Phi\left(\ketbra{i}{i}\right)=\sum_{k,l} \Phi_{k,l}^{i,i}\ketbra{k}{l}
	\end{align}
	and the condition $\Delta \Phi \Delta= \Phi \Delta$, we find
	\begin{align}
		\Phi_{k,l}^{i,i}=\Phi_{k,k}^{i,i} \delta_{k,l} \forall i.
	\end{align}
	Since $\Phi$ is a trace preserving quantum operation and $\ketbra{i}{i}$ a valid density operator, $\Phi_{k,k}^{i,i}$ have to be a non-negative real numbers with $\sum_k \Phi_{k,k}^{i,i}=1$. Thus $p(k|i)=\Phi_{k,k}^{i,i}$ defines conditional probabilities and we end up with the condition
	\begin{align}
		\Phi_{k,l}^{i,i}=&p(k|i) \delta_{k,l} \forall i,k,l.
	\end{align}
	This time, we need to enforce trace preservation separately, again by the condition
	\begin{align}
		\sum_a  \Phi_{a,a}^{b,d}=\delta_{b,d}.
	\end{align} Sufficiency is straight forward, the third statement a trivial consequence of the others.
\end{proof}

\section{Proofs of the results in the main text}
\setcounter{theorem}{1}
Here we give the proofs of our results in the main text, which we restate for readability.

\begin{proposition} \label{propA:mincPOVM}
	A POVM is free iff
	\begin{align}
		P_n= \sum_i P_i^n\ketbra{i}{i} \ \forall n.
	\end{align}
\end{proposition}
\begin{proof}
	Let us use the notation
	\begin{align}
		P_n=\sum_{i,j} P^n_{i,j} \ketbra{i}{j}
	\end{align}
	and
	\begin{align}
		\ket{\psi}=\frac{1}{\sqrt{2}}\left(\ket{a}+e^{i \phi}\ket{b}\right)
	\end{align}
	with $a\ne b$ and $\phi \in \rr$. Now assume that the POVM is free. From Def.~1 in the main text follows for $\tilde{\rho}=\ketbra{\psi}{\psi}$
	\begin{align}
		&\tr P_n \Delta \tilde{\rho} = \tr P_n \tilde{\rho} \nonumber \\
		\Leftrightarrow\ & P^n_{a,a} + P^n_{b,b} = P^n_{a,a} + P^n_{b,b} + P^n_{a,b} e^{-i\phi} + P^n_{b,a} e^{i \phi} \nonumber \\
		\Leftrightarrow\ &0= 2 \Real\left( P^n_{a,b} e^{-i\phi}\right)
	\end{align}
	such that
	\begin{align}
		P_n= \sum_i P_i^n\ketbra{i}{i} \ \forall n
	\end{align}
	is a necessary condition. It is obviously also sufficient.
\end{proof}

\setcounter{theorem}{6}
\begin{proposition}\label{propA:genMonotone}
	Let $\|\cdot\|$ denote a sub-multiplicative norm on quantum operations which is sub-multiplicative with respect to tensor products. If $\|\tilde{\Phi}\|\le 1$ for all $\tilde{\Phi}$ detection-incoherent, the functional
	\begin{align}
		M(\Theta)=\min_{\Phi \in \mathcal{DI}} \|\Delta \Theta -\Delta \Phi \|
	\end{align}
	is a measure in the detection-incoherent setting.
\end{proposition}
\begin{proof}
	Since $\|\cdot\|$ is a norm, $M(\Theta)$ is faithful. Remember that $\tilde{\Phi}\Phi\in \mathcal{DI}$ if both $\tilde{\Phi},\Phi \in \mathcal{DI}$. For $\tilde{\Phi} \in \mathcal{DI}$, we therefore have
	\begin{align}
		M(\Theta \tilde{\Phi})=&\min_{\Phi \in \mathcal{DI}} \|\Delta \Theta \tilde{\Phi} -\Delta \Phi \| \nonumber \\
		\le& \min_{\Phi \in \mathcal{DI}} \|\Delta \Theta \tilde{\Phi} -\Delta \Phi \tilde{\Phi} \| \nonumber \\
		\le& \min_{\Phi \in \mathcal{DI}} \|\Delta \Theta  -\Delta \Phi  \| \|\tilde{\Phi}\|\nonumber \\
		\le& \min_{\Phi \in \mathcal{DI}} \|\Delta \Theta  -\Delta \Phi  \| \nonumber \\
		=&M(\Theta)
	\end{align}
	and
	\begin{align}
		M( \tilde{\Phi} \Theta )=&\min_{\Phi \in \mathcal{DI}} \|\Delta \tilde{\Phi} \Theta  -\Delta \Phi \| \nonumber \\
		=&\min_{\Phi \in \mathcal{DI}} \|\Delta \tilde{\Phi} \Delta \Theta  -\Delta \Phi \| \nonumber \\
		\le& \min_{\Phi \in \mathcal{DI}} \|\Delta \tilde{\Phi} \Delta \Theta  -\Delta \tilde{\Phi} \Phi \| \nonumber \\
		=& \min_{\Phi \in \mathcal{DI}} \|\Delta \tilde{\Phi} \Delta \Theta  -\Delta \tilde{\Phi} \Delta \Phi  \| \nonumber \\
		\le& \min_{\Phi \in \mathcal{DI}} \|\Delta \tilde{\Phi}\| \|\Delta \Theta  -\Delta \Phi  \| \nonumber \\
		\le& \min_{\Phi \in \mathcal{DI}} \|\Delta \Theta  -\Delta \Phi  \| \nonumber \\
		=&M(\Theta).
	\end{align}
	Since the norm is sub-multiplicative with respect to tensor products, we have
	\begin{align}
		M(\Theta \otimes\id )=&\min_{\Phi \in \mathcal{DI}}\|\Delta (\Theta \otimes\id) -\Delta \Phi \| \nonumber \\
		\le&\min_{\Phi=\Phi_1\otimes\id \in \mathcal{DI}}\|(\Delta \Theta) \otimes\Delta -(\Delta \Phi_1 ) \otimes\Delta\| \nonumber \\
		=&\min_{\Phi_1 \in \mathcal{DI}}\|(\Delta \Theta -\Delta \Phi_1 ) \otimes\Delta\| \nonumber \\
		\le&\min_{\Phi_1 \in \mathcal{DI}}\|\Delta \Theta -\Delta \Phi_1 \|\|\Delta\| \nonumber \\
		\le&M(\Theta).
	\end{align}
	Convexity is a consequence of absolute homogeneity and the triangle inequality. Choose $\Phi_1$, $\Phi_2$ such that
	\begin{align}
		M(\Theta)=&\|\Delta \Theta-\Delta \Phi_1\|, \nonumber \\
		M(\Psi)=&\|\Delta \Psi-\Delta \Phi_2\|.
	\end{align}
	Then, for $0\le t\le 1$, we find
	\begin{align}
		M(t\Theta +(1-t) \Psi)=&\min_{\Phi \in \mathcal{DI}}\| \Delta (t\Theta +(1-t) \Psi)-\Delta \Phi \| \nonumber \\
		\le&\| \Delta (t\Theta +(1-t) \Psi)-\Delta (t\Phi_1+(1-t)\Phi_2)\| \nonumber \\
		=& \| t(\Delta \Theta -\Delta \Phi_1) +(1-t)(\Delta \Psi-\Delta \Phi_2) \| \nonumber \\
		\le& t \| \Delta \Theta -\Delta \Phi_1\| +(1-t)\|\Delta \Psi-\Delta \Phi_2 \| \nonumber \\
		=&tM(\Theta)+(1-t)M(\Psi).
	\end{align}
\end{proof}

\begin{theorem} \label{corA:DiamondMonotone}
	The functional
	\begin{align}
		M_\diamond(\Theta)=\min_{\Phi \in \mathcal{DI}} \|\Delta \Theta -\Delta \Phi \|_\diamond
	\end{align}
	is a measure in the detection-incoherent setting. We call this measure the diamond-measure.
\end{theorem}
\begin{proof}
	The diamond norm of a quantum operation is equal to one, the diamond norm is sub-multiplicative and sub-multiplicative with respect to tensor products~\cite{watrous2018theory}.
\end{proof}

\begin{theorem}
	The functional
	\begin{align}
		\tilde{M}_\diamond(\Theta)=\min_{\Phi\in \mathcal{DI}} \max_{\ket{\psi}}\|\Delta \left(\Theta- \Phi \right) \ketbra{\psi}{\psi} \|_1
	\end{align}
	is a measure in the detection-incoherent setting.
	We call it the nSID-measure (non-stochasticity in detection).
\end{theorem}
\begin{proof}
	Since the induced trace norm is not sub-multiplicative with respect to tensor products, we cannot use Prop.~\ref{propA:genMonotone}.
	To begin the proof, we notice that by convexity,
	\begin{align}
		\max_{\ket{\psi}}\|\Delta \left(\Theta- \Phi \right) \ketbra{\psi}{\psi} \|_1= \max_{\sigma}\|\Delta \left(\Theta- \Phi \right) \sigma \|_1.
	\end{align}
	For $\tilde{\Phi} \in \mathcal{DI}$, we have
	\begin{align}
		\tilde{M}_\diamond(\tilde{\Phi})=&\min_{\Phi\in \mathcal{DI}} \max_{\ket{\psi}}\|\Delta \left(\tilde{\Phi}- \Phi \right) \ketbra{\psi}{\psi} \|_1 \nonumber \\
		\le& \max_{\ket{\psi}}\|\Delta \left(\tilde{\Phi}- \tilde{\Phi}\right) \ketbra{\psi}{\psi} \|_1 \nonumber \\
		=&0
	\end{align}
	and for $\Delta \Theta \ne \Delta \Theta \Delta$, there exists for all  $\Phi \in \mathcal{DI}$ a $\ket{\psi}$ such that we have $(\Delta \Theta-\Delta \Phi \Delta) \ketbra{\psi}{\psi} \ne 0 $. Since $\| \cdot \|_1$ is a norm, this proves faithfulness.
	
	From this follows, again for $\tilde{\Phi} \in \mathcal{DI}$,
	\begin{align}
		\tilde{M }_\diamond(\Theta \tilde{\Phi})=&\min_{\Phi\in \mathcal{DI}} \max_{\sigma}\|\Delta \left(\Theta \tilde{\Phi}- \Phi \right) \sigma \|_1 \nonumber \\
		\le&\min_{\Phi\in \mathcal{DI}} \max_{\sigma}\|\Delta \left(\Theta\tilde{\Phi}- \Phi \tilde{\Phi} \right) \sigma \|_1 \nonumber \\
		=&\min_{\Phi\in \mathcal{DI}} \max_{\sigma}\|\Delta \left(\Theta- \Phi  \right) \tilde{\Phi} \sigma \|_1 \nonumber \\
		=&\min_{\Phi\in \mathcal{DI}} \max_{\rho=\tilde{\Phi}\sigma}\|\Delta \left(\Theta- \Phi  \right) \rho \|_1 \nonumber \\
		\le&\min_{\Phi\in \mathcal{DI}} \max_{\sigma}\|\Delta \left(\Theta- \Phi  \right) \sigma \|_1 \nonumber \\
		=&\tilde{M }_\diamond(\Theta ),
	\end{align}
	where we used in the second line that $\Phi \tilde{\Phi} \in \mathcal{DI}$ if $\Phi,\tilde{\Phi} \in \mathcal{DI}$.
	Using that the trace norm is contractive under CPTP maps, we find
	\begin{align}
		\tilde{M }_\diamond( \tilde{\Phi} \Theta)=&\min_{\Phi\in \mathcal{DI}} \max_{\sigma}\|\Delta \left(\tilde{\Phi}\Theta - \Phi \right) \sigma \|_1 \nonumber \\
		\le&\min_{\Phi\in \mathcal{DI}} \max_{\sigma}\|\Delta \left(\tilde{\Phi}\Theta - \tilde{\Phi} \Phi \right) \sigma \|_1 \nonumber \\
		=&\min_{\Phi\in \mathcal{DI}} \max_{\sigma}\| \Delta\tilde{\Phi} \Delta \left(\Theta - \Phi \right) \sigma \|_1 \nonumber \\
		\le&\min_{\Phi\in \mathcal{DI}} \max_{\sigma}\| \Delta \left(\Theta - \Phi \right) \sigma \|_1 \nonumber \\
		=&\tilde{M }_\diamond( \Theta).
	\end{align}
	With the help of Lem.~\ref{lem:evecM} follows
	\begin{align}
		\tilde{M }_\diamond( \Theta\otimes \id)=&\min_{\tilde{\Phi}\in \mathcal{DI}} \max_{\sigma}\left\|\Delta \left(\Theta \otimes \id - \tilde{\Phi} \right) \sigma \right\|_1 \nonumber \\
		\le&\min_{\Phi\in \mathcal{DI}} \max_{\sigma}\|\Delta \left(\Theta \otimes \id - \Phi  \otimes \id\right) \sigma \|_1 \nonumber \\
		=&\min_{\Phi\in \mathcal{DI}} \max_{\sigma}\|\left(\Delta \Theta\otimes \id  -\Delta  \Phi\otimes \id  \right)\left( \id \otimes \Delta\right) \sigma \|_1 \nonumber \\
		=&\min_{\Phi\in \mathcal{DI}} \max_{\rho=(\id\otimes\Delta)\sigma}\|\left(\Delta \Theta\otimes \id  -\Delta  \Phi\otimes \id  \right)\rho\|_1 \nonumber \\
		\le&\min_{\Phi\in \mathcal{DI}} \max_{q_{a|b},p_b,\ket{\phi_{a|b},b}}\left\|\sum_{a,b} q_{a|b} p_b \left(\Delta \Theta\otimes \id  -\Delta  \Phi\otimes \id  \right)\ketbra{\phi_{a|b},b}{\phi_{a|b},b} \right\|_1 \nonumber \\
		\le&\min_{\Phi\in \mathcal{DI}} \max_{q_{a|b},p_b,\ket{\phi_{a|b},b}} \sum_{a,b} q_{a|b} p_b \| \left(\Delta \Theta\otimes \id  -\Delta  \Phi\otimes \id  \right)\ketbra{\phi_{a|b},b}{\phi_{a|b},b} \|_1 \nonumber \\
		\le&\min_{\Phi\in \mathcal{DI}} \sum_{a,b} q_{a|b} p_b  \max_{\ket{\phi,i}} \| \left(\Delta \Theta\otimes \id  -\Delta  \Phi\otimes \id  \right)\ketbra{\phi,i}{\phi,i} \|_1 \nonumber \\
		=&\min_{\Phi\in \mathcal{DI}}  \max_{\ket{\phi},\ket{i}} \| \left[\left(\Delta \Theta  -\Delta  \Phi \right)\ketbra{\phi}{\phi}\right] \otimes\ketbra{i}{i} \|_1 \nonumber \\
		=&\min_{\Phi\in \mathcal{DI}}  \max_{\ket{\phi},\ket{i}} \| \left(\Delta \Theta  -\Delta  \Phi \right)\ketbra{\phi}{\phi} \|_1 \|\ketbra{i}{i} \|_1 \nonumber \\
		=&\min_{\Phi\in \mathcal{DI}}  \max_{\ket{\phi}} \| \left(\Delta \Theta  -\Delta  \Phi \right)\ketbra{\phi}{\phi} \|_1  \nonumber \\
		=&\tilde{M }_\diamond( \Theta).
	\end{align}
	Convexity is a consequence of absolute homogeneity and the triangle inequality. Choose $\Phi_1$, $\Phi_2$, $\sigma_1$, $\sigma_2$ such that
	\begin{align}
		\tilde{M}_\diamond(\Theta)=&\|\left(\Delta \Theta-\Delta \Phi_1 \right)\sigma_1\|_1 \nonumber \\
		\tilde{M}_\diamond(\Psi)=&\|\left(\Delta \Psi-\Delta \Phi_2 \right)\sigma_2\|_1.
	\end{align} Then, for $0\le t\le 1$, we find
	\begin{align}
		\tilde{M}_\diamond(t\Theta +(1-t) \Psi)=&\min_{\Phi \in \mathcal{DI}} \max_\sigma \|\left[ \Delta (t\Theta +(1-t) \Psi)-\Delta \Phi\right] \sigma \|_1 \nonumber \\
		\le&\max_\sigma\|\left[ \Delta (t\Theta +(1-t) \Psi)-\Delta (t\Phi_1+(1-t)\Phi_2)\right]\sigma \|_1 \nonumber \\
		=&\max_\sigma \| t(\Delta \Theta -\Delta \Phi_1)\sigma +(1-t)(\Delta \Psi-\Delta \Phi_2)\sigma \|_1 \nonumber \\
		\le&\max_\sigma \left[ t \| \left(\Delta \Theta -\Delta \Phi_1 \right) \sigma\|_1 +(1-t)\|\left(\Delta \Psi-\Delta \Phi_2 \right) \sigma \|_1 \right]\nonumber \\
		\le&\max_\sigma  t \| \left(\Delta \Theta -\Delta \Phi_1 \right) \sigma\|_1 + \max_\sigma (1-t)\|\left(\Delta \Psi-\Delta \Phi_2 \right) \sigma \|_1 \nonumber \\
		=&tM(\Theta)+(1-t)M(\Psi).
	\end{align}
\end{proof}

\begin{proposition}\label{propA:MaxMeas}
	The maximum value of $\tilde{M}_\diamond(\Theta)$ for $\Theta$ a quantum channel with input of dimension $n$ and output of dimension $m$ is given by
	\begin{align}
		\frac{2 (N_0-1)}{N_0},
	\end{align}
	where $N_0=\min\{n,m\}$. It is both saturated by a Fourier transform in a subspace of dimension $N_0$ and by a measurement in the Fourier basis, encoding the outcomes in the incoherent basis.
\end{proposition}
\begin{proof}
	We first prove the bound given in the proposition. We need to distinguish two cases. \\
	For $n\leq m$:
	\begin{align}
		\min_{\Phi\in \mathcal{DI} }\max_{\rho}& \left\| \Delta  (\Lambda-\Phi) \rho\right\|_1  \nonumber\\
		&\leq \min_{\tilde{\Phi}\in \mathcal{DI} }\max_{\rho}\left\| \Delta  (\Lambda-\Lambda \Delta \tilde{\Phi}) \rho\right\|_1  \nonumber\\
		&= \min_{\tilde{\Phi}\in \mathcal{DI} }\max_{\rho}\left\| \Delta \Lambda (\id_n-\Delta \tilde{\Phi}) \rho\right\|_1  \nonumber\\
		&\leq  \min_{\tilde{\Phi}\in \mathcal{DI} }\max_{\rho}\left\| (\id_n-\Delta \tilde{\Phi}) \rho\right\|_1  \nonumber\\
		&\leq 	\max_{\rho} \left\| \rho- \frac{\id_n}{n}\right\|_1 \nonumber\\
		&= \left\|\ketbra{1}{1}-\frac{\id_{n}}{n}\right\|_1 \nonumber\\
		&=(1-1/n)+(n-1)/n \nonumber\\
		&=\frac{2 (n-1)}{n}.
	\end{align}
	For $n\geq m$:
	\begin{align}
		\min_{\Phi\in \mathcal{DI} }\max_{\rho}& \left\| \Delta  (\Lambda-\Phi) \rho\right\|_1 \nonumber\\
		&\leq
		\max_{\rho} \left\| \Delta \Lambda \rho-\frac{\id_m}{m}\right\|_1 \nonumber\\
		&	\leq \left\|\ketbra{1}{1}-\frac{\id_{m}}{m}\right\|_1 \nonumber\\
		&=(1-1/m)+(m-1)/m \nonumber\\
		&=\frac{2 (m-1)}{m}.
	\end{align}
	That the Fourier transform $(FT)$, or the measurement in the Fourier basis, saturate the bound follows from the fact that inputting the $N_0$ states that are sent to the respective orthogonal incoherent states, one gets $\Delta \Phi \rho=\Delta \Phi\Delta_{N_0}  \rho=\Delta \Phi\frac{\id_{N_0}}{N_0}$, where we assumed without loss of generality that $FT$ acts non-trivially on the span of the first $N_0$ states and $\id_{N_0}$ denotes the identity on this space (and the first equality comes from $\Phi\in \mathcal{DI}$). Assuming that $\Phi$ does not act as the identity superoperator for these test-states, results in one of the respective resulting states having a bigger distance than $\frac{2 (N_0-1)}{N_0}$. Therefore the distance is at least given by $\frac{2 (N_0-1)}{N_0}$; i.e. the Fourier transform saturates the bound.
\end{proof}

\section{Semidefinite program for the diamond-measure.}\label{secA:SemiDia}
For a quantum operation $\Theta=\Theta^{B\leftarrow A}$, we define its corresponding Choi state~\cite{CHOI1975285,PhysRevA.87.022310} by
\begin{align}\label{eqA:Choi}
	J(\Theta)=\sum_{i,j}\Theta(\ketbra{i}{j})\otimes \ketbra{i}{j}.
\end{align}
The diamond-measure can be calculated efficiently using the semidefinite program	
\begin{alignat}{4}
	&\qquad \emph{Primal }	&& \emph{problem} 						&&\qquad  \quad \emph{Dual }&& \emph{problem} \nonumber \\
	&\text{minimize:}		&&2\| \tr_B(Z)\|_\infty					&& \text{maximize:}			&&2\left(\tr(J(\Delta\Theta )X)-\tr(Y_2)\right) \nonumber \\
	&\text{subject to:}		&&Z\ge J(\Delta\Theta)-W, 	\qquad		&& \text{subject to:}		&&X\le \id_B\otimes\rho  : \rho\ge0, \tr(\rho)=1,\nonumber \\
	& 						&&[\id-\Delta]W=0, 						&&							&&[\id-\Delta]Y_1-X+\id_B\otimes Y_2\ge 0, \nonumber \\
	& 						&&\tr_B(W)=\id_A,						&&							&&X\ge0, \nonumber \\
	& 						&&Z\ge 0,								&&							&&Y_1=Y_1^\dagger,	 \nonumber \\
	&						&&W \ge 0, 								&&							&&Y_2=Y_2^\dagger,	
\end{alignat}
which is based on \cite{v005a011}. Strong duality holds. Note that $\tr_B$ is the partial trace over the \textit{first} subsystem since $J(\Delta\Theta) \in B\otimes A$ (see Eq.~(\ref{eqA:Choi})).
\begin{proof}
	According to \cite{v005a011}, $\|\Delta \Theta -\Delta \Phi \|_\diamond$ is the optimal value of
	\begin{alignat}{2}
		&\text{minimize:}\quad	&&2\| \tr_B(Z)\|_\infty					\nonumber \\
		&\text{subject to:}		&&Z\ge J(\Delta\Theta-\Delta \Phi),	\nonumber \\
		& 						&&Z\ge 0.			
	\end{alignat}
	Therefore, $M_\diamond(\Theta)=\min_{\Phi \in \mathcal{DI}} \|\Delta \Theta -\Delta \Phi \|_\diamond$ is the optimal value of
	\begin{alignat}{2} \label{initialOpt}
		&\text{minimize:}\quad	&&2\| \tr_B(Z)\|_\infty					\nonumber \\
		&\text{subject to:}		&&Z\ge J(\Delta\Theta-\Delta \Phi),				\nonumber \\
		&						&&\Phi \in \mathcal{DI}, \nonumber \\
		& 						&&Z\ge 0.						
	\end{alignat}
	For $\Phi\in \mathcal{DI}$, we find
	\begin{align}
		J(\Delta \Phi)=	& J(\Delta \Phi \Delta) \nonumber \\
		=	& \sum_{i,k} \Phi_{k,k}^{i,i} \ketbra{k}{k}_B\otimes \ketbra{i}{i}_A
	\end{align}
	with $\Phi_{k,k}^{i,i}=p(k|i)$ according to Eq.~(\ref{m-inc cara}).  Thus (\ref{initialOpt}) is equivalent to
	\begin{alignat}{2}
		&\text{minimize:}\quad	&&2\| \tr_B(Z)\|_\infty					\nonumber \\
		&\text{subject to:}		&&Z\ge J(\Delta\Theta)-W,				\nonumber \\
		& 						&&[\id-\Delta]W=0, 						\nonumber \\
		& 						&&\tr_B(W)=\id_A,						\nonumber \\
		& 						&&Z\ge 0, 								\nonumber \\
		&						&&W\ge 0,
	\end{alignat}
	which is the primal problem.
	This can be reformulated as
	\begin{alignat}{2}
		&\text{minimize:}\quad	&&a 									\nonumber \\
		&\text{subject to:}		&&a\id_A-2\tr_B(Z)\ge 0,				\nonumber \\
		&						&&Z\ge J(\Delta\Theta)-W,				\nonumber \\
		& 						&&[\id-\Delta]W=0, 						\nonumber \\
		& 						&&\tr_B(W)=\id_A,						\nonumber \\
		& 						&&Z\ge 0, 								\nonumber \\
		&						&&W\ge 0,								\nonumber \\
		&						&&a\ge0.
	\end{alignat}
	The corresponding Lagrangian is given by
	\begin{align}
		L\left(a,Z,W,\tilde{X},X,Y_1,Y_2\right)=& a+\tr\left(\left(2\tr_B Z -a \id_A\right)\tilde{X}\right)+\tr\left(\left(J\left( \Delta \Theta\right)-W-Z\right)X\right) \nonumber \\
		&+ \tr \left( \left[\id -\Delta\right](W)Y_1\right)+\tr\left(\left(\tr_B W -\id_A\right) Y_2\right)
	\end{align}
	and the dual function by
	\begin{align}
		q\left(\tilde{X},X,Y_1,Y_2 \right)=&\inf_{a,Z,W\ge0} L\left(a,Z,W,\tilde{X},X,Y_1,Y_2\right) \nonumber \\
		=&\inf_{a,Z,W\ge0} \tr\left(J\left( \Delta \Theta\right)X\right)-\tr\left(Y_2\right)+a\left(1-\tr \tilde{X}\right)+2\tr\left(\tr_B\left(Z\right)\tilde{X}\right)\nonumber \\
		&-\tr\left(ZX\right)+\tr\left(W Y_1\right) -\tr\left(\Delta\left[W\right]Y_1\right)- \tr\left(W X\right)+\tr\left(\tr_B\left(W\right)Y_2\right).
	\end{align}
	With
	\begin{align}
		\tr\left(\tr_B\left(Z\right)\tilde{X}\right)=\tr\left(Z\left(\id_B\otimes\tilde{X}\right)\right)		
	\end{align}
	and
	\begin{align}
		\tr\left(\Delta\left[W\right]Y_1\right)=\tr\left(W\Delta\left[Y_1\right]\right)
	\end{align}
	follows
	\begin{align}
		q\left(\tilde{X},X,Y_1,Y_2 \right)=&\inf_{a,Z,W\ge0} \tr\left(J\left( \Delta \Theta\right)X\right)-\tr\left(Y_2\right)+ a\left(1-\tr \tilde{X}\right)\nonumber \\
		&+\tr \left(Z\left(2\id_B\otimes\tilde{X}-X\right)\right)+\tr\left(W\left(Y_1-\Delta Y_1 -X+\id_B\otimes Y_2\right)\right) \nonumber \\
		=&\begin{cases} \tr\left(J\left( \Delta \Theta\right)X\right)-\tr\left(Y_2\right) &\text{if } \tr\tilde{X}\le 1 \land 2 \id_B\otimes\tilde{X}-X\ge 0 \land Y_1-\Delta Y_1-X+\id_B\otimes Y_2 \ge 0, \\
			-\infty &\text{else. } \end{cases}
	\end{align}
	Thus the dual problem is given by
	\begin{alignat}{2}
		&\text{maximize:}\quad	&&\tr\left(J\left( \Delta \Theta\right)X\right)-\tr\left(Y_2\right)									\nonumber \\
		&\text{subject to:}		&&\tr \tilde{X}\le 1,				\nonumber \\
		&						&&2 \id_B\otimes \tilde{X} -X\ge 0,  				\nonumber \\
		& 						&&Y_1-\Delta Y_1-X+\id_B\otimes Y_2 \ge 0,			\nonumber \\
		& 						&&\tilde{X}\ge0,						\nonumber \\
		& 						&&X\ge 0, 								\nonumber \\
		&						&&Y_1=Y_1^\dagger,					\nonumber \\
		&						&&Y_2=Y_2^\dagger.		
	\end{alignat}
	Assume $\tilde{X}\ge 0$ and $\tr\tilde{X}<1$. Then $\tilde{X}':=\frac{1}{\tr\tilde{X}}\tilde{X}=\left(1+c\right)\tilde{X}$ has trace one, is positive semidefinite and
	\begin{align}
		2\id_B\otimes\tilde{X}'-X=2\id_B\otimes\tilde{X} -X+2c\id_B\otimes\tilde{X}
	\end{align}
	is positive semidefinite for all $X$ that satisfy $2 \tilde{X} \otimes \id_B-X\ge 0$. Thus we can simplify the dual problem to
	\begin{alignat}{2}
		&\text{maximize:}\quad	&&\tr\left(J\left( \Delta \Theta\right)X\right)-\tr\left(Y_2\right)									\nonumber \\
		&\text{subject to:}		&&X \le 2 \id_B\otimes\rho  : \rho\ge0, \tr(\rho)=1,			\nonumber \\
		& 						&&Y_1-\Delta Y_1-X+ \id_B \otimes Y_2\ge 0,			\nonumber \\
		& 						&&X\ge 0, 								\nonumber \\
		&						&&Y_1=Y_1^\dagger,					\nonumber \\
		&						&&Y_2=Y_2^\dagger.	
	\end{alignat}
	Finally we can define $X'=\frac{1}{2}X$, $Y_1'=\frac{1}{2}X$ and $Y_2'=\frac{1}{2}X$ to arrive at the dual problem stated.
	
	To show that strong duality holds, we write the primal problem as
	\begin{alignat}{2}
		&\text{minimize:}\quad	&&2\| \tr_B(Z)\|_\infty					\nonumber \\
		&\text{subject to:}		&&J(\Delta\Theta)-W-Z\le 0,			\nonumber \\
		& 						&&-Z\le 0, 								\nonumber \\
		&						&&-W\le 0,									\nonumber \\
		& 						&&[\id-\Delta]W=0, 						\nonumber \\
		& 						&&\tr_B(W)=\id_A.						
	\end{alignat}
	According to \cite{boyd2004convex}, strong duality holds if there exist $Z',W'$ such that the equality constraints are satisfied and the inequality constraints are strictly satisfied. If we choose
	\begin{align}
		Z'=\id_B\otimes \id_A+J(\Delta \Theta), \nonumber \\
		W'=\frac{1}{\dim B}\id_B\otimes \id_A,		
	\end{align}
	this is obviously the case.
\end{proof}

\section{Operational interpretation of the nSID-measure}\label{secA:OpIntBias}
In the following, we complete the proof of the operational interpretation of the nSID-measure given in the main text.
What remains to show is the identity
\begin{align}
	P_c(1/2,\Theta_0,\Theta_1)=& \frac{1}{2}+ \frac{1}{4} \max_{\ket{\psi}} \|\Delta\left(\Theta_0-\Theta_1\right) \ketbra{\psi}{\psi} \|_1.
\end{align}
First we show the following proposition, which is a special case of results in~\cite{Matthews2009}. For completeness, we give a direct proof.
\setcounter{theorem}{15}
\begin{proposition}
	Assume you obtain a single copy of a quantum state which is with probability $\lambda$ equal to $\rho_0$ and with probability $1-\lambda$ equal to $\rho_1$. The optimal probability $P_c(\lambda,\rho_0,\rho_1)$ to correctly guess $i=0,1$ when one can perform only incoherent measurements is given by
	\begin{align}
		P_c(\lambda,\rho_0,\rho_1)=& \frac{1}{2}+\frac{1}{2} \| \Delta \left(\lambda \rho_0-(1-\lambda)\rho_1\right)\|_1. \nonumber
	\end{align}
\end{proposition}

\begin{proof}
	The optimal strategy to correctly guess $i$ is based on the outcome of an dichotomic POVM $\{P_0,P_1=\id-P_0\}$ in the set of allowed measurements where we guess $i$ whenever we measured $i$. This can be seen by the following arguments: In the end, we have to make a dichotomic guess. This can only be based on a the outcomes of an (not necessarily dichotomic) incoherent POVM. In principle, we could post-process the measurement outcomes in a stochastic manner to arrive at our dichotomic guess. However, this stochastic post-processing can be incorporated into the definition of a new incoherent POVM. In addition, an optimal strategy includes the usage of all information obtainable, therefore the task consists in finding an optimal $P_0$.
	
	Let us define
	\begin{align}
		\rho=\lambda \rho_0+(1-\lambda) \rho _1, \nonumber \\
		X= \lambda \rho_0-(1-\lambda)\rho_1
	\end{align}
	and $\mathcal{I}=\left\{ i: X_{i,i}=\bra{i}X\ket{i}>0 \right\}$.
	For a fixed and not necessarily optimal $P_0$, the probability $P_c(P_0;\lambda,\rho_0,\rho_1)$ to guess correctly is then given by
	\begin{align}
		P_c(P_0;\lambda,\rho_0,\rho_1)=&\lambda \tr \left(P_0\rho_0\right)+(1-\lambda)\tr \left(P_1 \rho_1\right) \nonumber \\
		=&\tr \left(P_0 \frac{\rho+X}{2}\right) +\tr \left(P_1 \frac{\rho-X}{2}\right) \nonumber \\
		=& \tr \left[ \left(P_0+P_1\right)\frac{\rho}{2}+\left(P_0-P_1\right)\frac{X}{2}\right] \nonumber \\
		=& \frac{1}{2}+ \tr \left[ \left(2P_0-\id\right)\frac{X}{2}\right] \nonumber \\
		=& \frac{1}{2}+ \tr \left[P_0X \right] -\frac{1}{2}(\lambda-(1-\lambda))\nonumber \\
		=& (1-\lambda)+ \sum_i P_i^0 X_{i,i}
	\end{align}
	and
	\begin{align}
		P_c(\lambda,\rho_0,\rho_1) =&  \max_{P_0} P_c(P_0;\lambda,\rho_0,\rho_1) \nonumber \\
		=&(1-\lambda)+\max_{P_0}  \sum_i P_i^0 X_{i,i} \nonumber \\
		=&(1-\lambda) + \sum_{i\in \mathcal{I}} X_{i,i}.
	\end{align}
	In addition,
	\begin{align}
		\frac{1}{2}+\frac{1}{2} \| \Delta \left(\lambda \rho_0-(1-\lambda)\rho_1\right)\|_1 =&  \frac{1}{2}+\frac{1}{2} \left\| \sum_i \left(\lambda \rho_{i,i}^0-(1-\lambda)\rho_{i,i}^0\right)\ketbra{i}{i} \right\|_1 \nonumber \\
		=&  \frac{1}{2}+\frac{1}{2}  \sum_i \left| \lambda \rho_{i,i}^0-(1-\lambda)\rho_{i,i}^0\right| \nonumber \\
		=&  \frac{1}{2}+\frac{1}{2}  \sum_i \left| X_{i,i}\right| \nonumber \\
		=&  \frac{1}{2}+\frac{1}{2} \left[ \sum_{i\in \mathcal{I}} X_{i,i} -\sum_{i\in \mathcal{I}^c} X_{i,i}  \right]\nonumber \\
		=&  \frac{1}{2}+\frac{1}{2} \left[ 2\sum_{i\in \mathcal{I}} X_{i,i} -\sum_{i} X_{i,i}  \right]\nonumber \\
		=&  \frac{1}{2}+\frac{1}{2} \left[ 2\sum_{i\in \mathcal{I}} X_{i,i} -(2\lambda-1)  \right]\nonumber \\
		=&(1-\lambda) + \sum_{i\in \mathcal{I}} X_{i,i},
	\end{align}
	which finishes the proof.
\end{proof}

This allows us to obtain a slightly more general result than needed.

\begin{proposition}
	Assume you obtain a single copy of a quantum channel which is with probability $\lambda$ equal to $\Theta_0$ and with probability $1-\lambda$ equal to $\Theta_1$. The optimal probability $P_c(\lambda,\Theta_0,\Theta_1)$ to correctly guess $i=0,1$ if one can perform only incoherent measurements is given by
	\begin{align}
		P_c(\lambda,\Theta_0,\Theta_1)=& \frac{1}{2}+ \frac{1}{2} \max_{\ket{\psi}} \|T \ketbra{\psi}{\psi} \|_1 \nonumber \\
	\end{align}
	for
	\begin{align}
		T=\Delta\left[\lambda\Theta_0-\left(1-\lambda\right)\Theta_1\right].
	\end{align}
\end{proposition}

\begin{proof}
	The optimal probability to guess correctly is given by the optimal probability to distinguish $\sigma_0=\left(\Theta_0\otimes \id\right)\sigma$  and $\sigma_1=\left(\Theta_1\otimes \id\right)\sigma$ for optimal $\sigma$ (note that this includes the strategy of applying $\Theta_i \otimes \mathcal{E}$ to $\sigma$ for an quantum channel $\mathcal{E}$). Therefore we have
	\begin{align}
		P_c(\lambda,\Theta_0,\Theta_1)=& \max_\sigma \frac{1}{2}+\frac{1}{2}\| \Delta \left[\lambda\Theta_0\otimes \id-(1-\lambda)\Theta_1\otimes \id\right]\sigma \|_1.
	\end{align}
	However, using Lem.~\ref{lem:evecM}
	\begin{align}
		\max_\sigma \| \Delta \left[\lambda\Theta_0\otimes \id-(1-\lambda)\Theta_1\otimes \id\right]\sigma \|_1=& \max_\sigma \| \left(T\otimes \Delta\right)\sigma \|_1 \nonumber \\
		=&\max_\sigma  \| \left(T\otimes\id\right) \left(\id \otimes \Delta\right)\sigma \|_1 \nonumber \\
		=&\max_{\rho=(\id\otimes\Delta)\sigma} \| \left(T\otimes\id\right) \rho \|_1 \nonumber \\
		=&\max \left\| \sum_{a,b} q_{a|b} p_b \left(T\otimes\id\right) \ketbra{\phi_{a|b},b}{\phi_{a|b},b} \right\|_1 \nonumber \\
		\le&\max \sum_{a,b} q_{a|b} p_b \|  \left(T\otimes\id\right) \ketbra{\phi_{a|b},b}{\phi_{a|b},b} \|_1 \nonumber \\
		\le& \sum_{a,b} q_{a|b} p_b \max_{\ket{\phi}, \ket{i}} \|  \left(T\otimes\id\right) \ketbra{\phi}{\phi} \otimes \ketbra{i}{i} \|_1 \nonumber \\
		=&\max_{\ket{\phi}, \ket{i}} \|  T \ketbra{\phi}{\phi} \|_1 \ \|\ketbra{i}{i} \|_1 \nonumber \\
		=&\max_{\ket{\phi}} \|  T \ketbra{\phi}{\phi} \|_1.
	\end{align}
\end{proof}

\section{Evaluating the nSID-measure}\label{secA:eval-bias}
In this section, we show how the nSID-measure can be evaluated. We first give an overview of the main ideas and steps.
From the definition of the trace norm in the main text and its convexity follows directly that
\begin{align}\label{SID}
	\tilde{M}_\diamond(\Theta)=\min_{\Phi \in \mathcal{DI}} \max_{\rho} \tr \left|	\Delta (\Theta-\Phi)\rho 	\right|,
\end{align}
where the maximization is over mixed states.
The idea is to separate this into an inner and an outer program and to use them iteratively. The outer program is given by
\begin{align}
	\min_{\Phi \in \mathcal{DI}}\max_{\rho \in D(n)}
	\tr \left|
	\Delta (\Theta-\Phi)\rho
	\right|,
\end{align}
where $D(n)$ is a discrete set of states. Due to this discreteness, we can rephrase this program as a linear program. Since the only difference to the original program is that we discretized the set of states, its optimal value yields a lower bound to the optimal value of~(\ref{SID}). Using the optimal $\Phi_n$ which achieves this bound, we can then calculate a $\rho_n$ as the optimal point of the inner program
\begin{align}
	\max_{\rho} \tr \left|	\Delta (\Theta-\Phi_n)\rho 	\right|.
\end{align}
This problem can be formulated as a branch and bound problem with semidefinite branches. Since the only difference to the original program is that we do not minimize over $\Phi$, we get an upper bound to the optimal value of~(\ref{SID}). Adding to $D(n)$ a basis which is obtained by rotating $\rho_n$ around the incoherent axis, we obtain $D(n+1)$.
This new set is used as the input for the next iteration.
Once the bounds coincide to the required accuracy, the problem is solved. In the following, we present the missing details.

First we formulate the program for the upper bound,
\begin{align}\label{Inner}
	\max_{\rho }
	\tr \left|
	\Delta (\Theta-\Phi_n)\rho
	\right|,
\end{align}
as a branch and bound problem.
Using the short hand notation $J_n=J(\Delta (\Theta-\Phi_n))$, we will show that the optimal value of the above optimization problem is equivalent to the optimal value of
\begin{alignat}{2} \label{final_inner}
	&\text{minimize:}\qquad 	&& -2 \tr[X J_n] \nonumber \\
	&\text{subject to:}			&&X=\bigoplus_i \rho_i,\nonumber\\
	&							&&0\leq\rho_i\leq \rho, \nonumber \\
	&							&&\tr[\rho_i]=B(i),\nonumber\\
	& 							&&\rho\geq 0, \nonumber\\
	&							&&\tr(\rho)=1,\nonumber\\
	&							&&B(i) \in \{0,1\}.
\end{alignat}
Note that for fixed $B$, this is a semidefinite program. We thus just need to minimize the different programs over the possible choices of $B$.

We first show now that~(\ref{Inner}) is equivalent to
\begin{alignat}{2}\label{inter_upper}
	&\text{minimize:}\qquad	&& -2 \tr[P \tr_2[(\id\otimes\rho) J_n] ] \nonumber \\
	&\text{subject to:} 	&&\rho\geq 0, \nonumber \\
	&						&&\tr(\rho)=1,\nonumber\\
	&						&&P^2=P,\nonumber \\
	&						&&P\geq 0,
\end{alignat}
where the last line means that $P$ is a projector, and the minimization implies that the optimal $P_0$ for any fixed state $\rho_0$ is the projector onto the positive part of $\tr_2[(\id\otimes\rho_0) J_n]=(\Delta (\Theta-\Phi_n))[\rho_0]$. Now we note that
\begin{align}
	\tr[P_0 (\Delta (\Theta-\Phi_n))[\rho_0] ]=-\tr[(\id -P_0) (\Delta (\Theta-\Phi_n))[\rho_0] ],
\end{align}
since $(\Delta (\Theta-\Phi_n))$ is the difference of two trace preserving maps. This implies that
\begin{align}
	2 \tr[P_0 \tr_2[(\id\otimes\rho_0) J_n] ]= \tr \left|
	(\Delta (\Theta-\Phi_n))[\rho_0]
	\right|,
\end{align} for any fixed state $\rho_0$, which gives the equivalence to~(\ref{Inner}).

Since $(\Delta (\Theta-\Phi_n))[\rho_0]=\tr_2[(\id\otimes\rho_0) J_n]$ is diagonal in the incoherent basis, the optimal $P_0$ for~(\ref{inter_upper}) is diagonal as well. Then we can restrict the optimization to these $P$. But diagonal $P$ can be rewritten as $P=\diag(B)$, with $B$ a vector with components either $0$ or $1$. The next step is to write
\begin{align}
	\tr[P \tr_2[(\id\otimes\rho) J_n]]=\tr[(P\otimes\id)(\id\otimes\rho) J_n]=\tr[(P\otimes\rho)J_n]=\tr[(\diag(B)\otimes\rho)J_n]=\tr[(\oplus_i B(i)\rho ) J_n].
\end{align} This means that~(\ref{Inner}) is equivalent to the problem
\begin{alignat}{2}
	&\text{minimize:}\qquad	&& -2 \tr[X J_n]  \nonumber\\
	&\text{subject to:}		&&X=\bigoplus_i B(i) \rho, \nonumber\\
	&						&&\rho\geq 0,\nonumber \\
	&						&&\tr(\rho)=1,\nonumber\\
	&						&&B(i) \in \{0,1\}
\end{alignat}
or
\begin{alignat}{2}\label{inter_upper2}
	&\text{minimize:}\qquad		&& -2 \tr[X J_n] \nonumber \\
	&\text{subject to:}			&&X=\bigoplus_i B(i) \tilde{\rho_i}, \nonumber\\
	&							&&B(i) \tilde{\rho_i}= B(i) \rho, \nonumber \\
	&							&&\tr[B(i)\tilde{\rho_i}]=B(i), \nonumber\\
	&							&&\rho\geq 0, \nonumber \\
	&							&&\tr(\rho)=1,\nonumber\\
	&							&&B(i) \in \{0,1\}.
\end{alignat}
Note that the constraint $\tr[B(i)\tilde{\rho_i}]=B(i)$ in the above program is always satisfied if the other constraints hold.
We arrive at~(\ref{final_inner}) by defining $B(i)\tilde{\rho_i}=\rho_i$  and relaxing the constraint $\rho_i= B(i)\rho$ to $0\leq\rho_i\leq \rho$. On the other hand, given the constraints of~(\ref{final_inner}) are valid, if $B(i)=0$, from $0\leq\rho_i$ and $\tr[\rho_i]=B(i)=0$ it follows that $\rho_i=0$ and therefore $\rho_i=0= B(i) \rho$. If $B(i)=1$, $\rho-\rho_i$ is a traceless hermitian operator and by the condition $\rho_i\leq \rho$ it is also positive. Hence it is zero and $\rho_i= \rho=B(i) \rho$. This proves the equivalence of~(\ref{final_inner}) to~(\ref{inter_upper2}) and finally to~(\ref{Inner}).

Next we formulate the linear program for the lower bound.
The outer problem, giving the lower bound, is given by
\begin{alignat}{2}\label{outer}
	&\text{minimize:}\qquad&& \max_{\rho_i \in D(n)}
	\tr \left|
	(\Delta (\Theta-\Phi))[\rho_i]
	\right| \nonumber \\
	&\text{subject to:}&&\Phi\in \mathcal{DI}.
\end{alignat}
First we note that $\Phi\in \mathcal{DI}$ implies that $S:=J(\Delta \Phi)$
is diagonal and therefore only defined by the transition probabilities $p(k|l)$ from the populations of the input states to the ones of the output states. Secondly we can calculate $\sigma_i=\Delta \Theta [\rho_i]$ for each $\rho_i \in D(n)$. Then we see that the only quantities that matter are the diagonal elements $r_i$ of $\rho_i$, and $s_i$ of $\sigma_i$ and we get the program
\begin{alignat}{2}\label{outer_intermediate}
	&\text{minimize:} \qquad&& \max_{i}
	\sum_k \left|
	s_i(k) - \sum_l p(k|l) r_i(l)
	\right| \nonumber \\
	&\text{subject to:} && \sum_k p(k|l)=1\forall l,\nonumber\\
	&					&&p(k|l)\geq 0\forall k,\,l.
\end{alignat}
Since both $\Delta\Theta$ and $\Delta\Phi$ are trace preserving operations,
\begin{align}
	\tr \left| (\Delta (\Theta-\Phi))[\rho_i] \right|=2 \tr \left( \mathrm{Pos}\left((\Delta (\Theta-\Phi)\right)[\rho_i]) \right)
\end{align}
where $\mathrm{Pos}$ denotes the positive part of an operator. From this follows that~(\ref{outer_intermediate}) is equivalent to
\begin{alignat}{3}\label{outer_final}
	&\text{minimize:}\qquad	&& 2 x \nonumber \\
	&\text{subject to:} 	&& x\geq \sum_k T_{k,i} 						&&\forall i,\nonumber\\
	&						&&T_{k,i}\geq 0 								&&\forall i,\,k,\nonumber\\
	&						&& T_{k,i}\geq s_i(k) - \sum_l p(k|l) r_i(l) 	&&\forall i,\,k,\nonumber\\
	&						&&\sum_k p(k|l)=1								&&\forall l,\nonumber\\
	&						&&p(k|l)\geq 0									&&\forall k,\,l.
\end{alignat}
For $N=|D(n)|$ and $\rho_i\in D(n)$, introducing the matrices $R\in \mathbb{R}^{d_{\text{in}}\times N }$ by $R(:,i)=\diag(\rho_i)$,
and $S\in \mathbb{R}^{d_{\text{out}}\times N }$ by $S(:,i)=\diag(\Delta\Theta[\rho_i])$ leads to
\begin{alignat}{3}\label{outer_matrix}
	&\text{minimize:}\qquad	&& 2x \nonumber \\
	&\text{subject to:} 	&& x\geq \sum_kT_{k,i} 							&&\forall i,\nonumber\\
	&						&&T_{k,i}\geq 0 								&&\forall i,\,k,\nonumber\\
	&						&& T_{k,i}\geq S_{k,i} - (P R)_{k,i} 			&&\forall i,\,k,\nonumber\\
	&						&&\sum_k P_{k,l}=1								&&\forall l,\nonumber\\
	&						&&P_{k,l}\geq 0									&&\forall k,\,l,\nonumber\\
	&						&&T\in \mathbb{R}^{d_{\text{out}}\times N },\nonumber\\
	&						&& P\in \mathbb{R}^{d_{\text{out}}\times d_{\text{in}} }.
\end{alignat}

\section{Examples}\label{secA:examples}
In this section, we calculate the measures introduced in the main text for two quantum channels acting on qutrits.
For the first example, we mix the total dephasing operation $\Delta$ from the main text, which is free, with the quantum Fourier transformation $FT$, which is most valuable according to the nSID-measure (see Prop.~\ref{propA:MaxMeas}). For $0\le p\le1$, we denote the resulting map by $\Theta$,
\begin{align}
	\Theta(\rho)=(1-p)\Delta\rho+p\ FT \rho.
\end{align}
Since $\Theta$ is free for $p=0$, both measures are zero in this case. For $p\ne0$, the operation is non-free, leading to non-zero measures. This is shown in Fig.~\ref{fig:Meas}, where it is also shown that in the case of $\Theta$, the two measures are equal within numerical precision. To show that this is not always the case, we present a second example, which is given by
\begin{align}
	\Lambda(\rho)=(1-p)\Delta\rho+p \sum_{n=1}^{3}K_n\rho K_n^\dagger,
\end{align}
where again $0\le p\le1$ and
\begin{align}
	K_{1}  =\frac{1}{\sqrt{4}}\begin{pmatrix}
		-1 & 1 & 0\\
		0 & 0 & 0 \\
		1 & 1 & 0 \\
	\end{pmatrix},
	K_{2}=\frac{1}{\sqrt{4}}\begin{pmatrix}1 & 0 & -1\\
		1 & 0 & 1 \\
		0 & 0 & 0 \\
	\end{pmatrix},
	K_{3}  =\frac{1}{\sqrt{4}}
	\begin{pmatrix}0 & -1 & 1\\
		0 & 0 & 0 \\
		0 & 1 & 1 \\
	\end{pmatrix}.
\end{align}
As shown in Fig.~\ref{fig:Meas}, the two measures are different for $\Lambda$ and the diamond-measure is, for given $p$, larger than the nSID-measure. In general, as can be seen directly form the definitions, the diamond-measure is an upper bound to the nSID-measure.
As expected from Prop.~\ref{propA:MaxMeas}, $\Theta$ is more "valuable" than $\Lambda$ for the same $p$.

\begin{figure}[h!]
	\centering \includegraphics[width=0.7\linewidth]{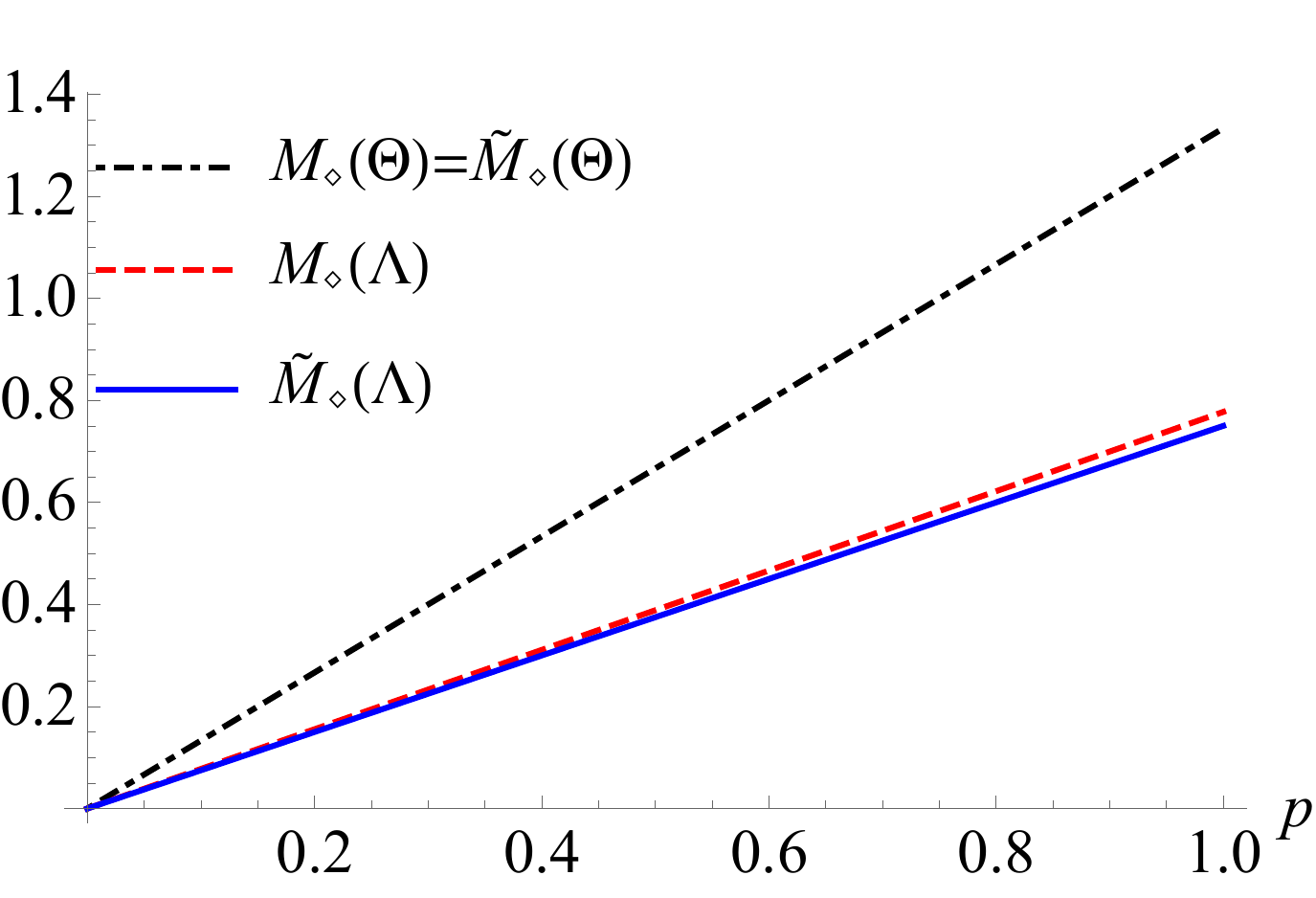} \caption{In this figure, the diamond-measure $M_\diamond$ and the nSID-measure $\tilde{M}_\diamond$ are plotted for the two exemplary quantum operations $\Theta$ and $\Lambda$.}
	\label{fig:Meas}
\end{figure}

\section{A measure for detection-creation incoherent operations} \label{secA:CreBoth}
Finally we want to give at least one example for a measure in the detection-creation-incoherent setting, which has been introduced in \cite{PhysRevLett.110.070502}. Therefore we denote by $S(\cdot \| \cdot)$ the (quantum) relative entropy.
\begin{proposition} \label{prop:DattaMeasure}
	The functional $M_c$ defined as \cite{PhysRevLett.110.070502}
	\begin{align}
		M_c\left(\Theta\right)= \sup_\rho S\Big(\Theta \Delta\rho \Big\| \Delta \Theta \rho\Big)
	\end{align}
	is a measure in the detection-creation-incoherent setting.
\end{proposition}
\begin{proof}
	Monotonicity under left and right composition is shown in \cite{PhysRevLett.110.070502} and faithfulness is given by the property of the relative entropy that $S(\rho \| \sigma)=0$ iff $\rho=\sigma$.
	From Lem.~\ref{lem:evecM}, we know that the eigenvectors of $\id^A \otimes \Delta^B (\rho^{A,B})$ are separable.
	For $\Gamma_b(\rho)=\rho\otimes \ketbra{b}{b}$ and using joint convexity and contractivity, we can then prove
	\begin{align}
		M_c\left(\Theta\otimes \id \right)=&\sup_\rho S\Big((\Theta\otimes\id)\Delta\rho\Big\| \Delta(\Theta\otimes\id) \rho\Big) \nonumber \\
		=&\sup_\rho S\Big(((\Theta \Delta)\otimes\id) (\id\otimes\Delta)\rho\Big\| ((\Delta\Theta)\otimes\id) (\id\otimes\Delta)\rho\Big) \nonumber \\
		=&\sup_{\sigma=(\id\otimes \Delta) \rho} S\Big(((\Theta \Delta)\otimes\id) \sigma\Big\| ((\Delta\Theta)\otimes\id)\sigma\Big) \nonumber \\
		=&\sup_{{q_{a|b},p_b,\ket{\phi_{a|b},b}}} S\left(((\Theta \Delta)\otimes\id) \sum_{a,b} q_{a|b} p_b \ketbra{\phi_{a|b},b}{\phi_{a|b},b}\Bigg\| ((\Delta \Theta)\otimes\id)\sum_{a,b} q_{a|b} p_b \ketbra{\phi_{a|b},b}{\phi_{a|b},b}\right) \nonumber \\
		\le& \sup_{{q_{a|b},p_b,\ket{\phi_{a|b},b}}} \sum_{a,b} q_{a|b} p_b  S\Big(((\Theta \Delta)\otimes\id) \ketbra{\phi_{a|b},b}{\phi_{a|b},b}\Big\| ((\Delta \Theta)\otimes\id)  \ketbra{\phi_{a|b},b}{\phi_{a|b},b}\Big)\nonumber \\
		=& \sup_{\ket{\phi,b}}   S\Big(((\Theta \Delta)\otimes\id)  \ketbra{\phi,b}{\phi,b} \Big\| ((\Delta\Theta)\otimes\id)  \ketbra{\phi,b}{\phi,b}\Big) \nonumber \\
		=& \sup_{\ket{\phi},b}   S\Big(\Gamma_b\Theta \Delta  \ketbra{\phi}{\phi}\Big\| \Gamma_b\Delta \Theta  \ketbra{\phi}{\phi}\Big) \nonumber \\
		\le& \sup_{\ket{\phi}}   S\Big(\Theta \Delta  \ketbra{\phi}{\phi}\Big\| \Delta\Theta  \ketbra{\phi}{\phi}\Big)\nonumber \\
		\le& \sup_{\rho}   S\Big(\Theta \Delta \rho\Big\| \Delta\Theta  \rho\Big)\nonumber \\
		=&M_c(\Theta).
	\end{align}
	
	In addition, we have
	\begin{align}
		M_c(\Theta)=&\sup_\rho S\Big(\Theta  \Delta\rho\Big\| \Delta \Theta \rho\Big) \nonumber \\
		=&\sup_{\rho  } S\Big(\tr_B \left(((\Theta     \Delta  )\otimes \id  )\left(\rho  \otimes \ketbra{1}{1}  \right)\right)\Big\| \tr_B \left( ((\Delta   \Theta  ) \otimes \id  )\left(\rho  \otimes \ketbra{1}{1}  \right)\right)\Big)\nonumber \\
		\le&\sup_{\rho  } S\Big(\left(\left(\Theta     \Delta  \right)\otimes \id  \right)\left(\rho  \otimes \ketbra{1}{1}  \right)\Big\|\left(\left(\Delta   \Theta  \right) \otimes \id  \right) \left(\rho  \otimes \ketbra{1}{1}  \right)\Big)\nonumber \\
		=&\sup_{\rho  } S\Big(\left(\Theta   \otimes \id  \right)   \Delta  \left(\rho  \otimes \ketbra{1}{1}  \right)\Big\|\Delta   \left(\Theta   \otimes \id  \right) \left(\rho  \otimes \ketbra{1}{1}  \right)\Big) \nonumber \\
		\le&\sup_{\rho  } S\Big((\Theta   \otimes \id )    \Delta  \rho  \Big\|\Delta   ( \Theta   \otimes \id ) \rho  \Big) \nonumber \\
		=&M_c\left(\Theta\otimes \id\right)
	\end{align}
	and thus
	\begin{align}
		M_c(\Theta)=M_c(\Theta\otimes \id).
	\end{align}
	Convexity is given by
	\begin{align}
		M_c\left(\sum_i q_i \Theta_i\right)=&\sup_\rho S\Bigg(\sum_i q_i \Theta_i    \Delta \rho\Bigg\| \sum_i q_i  \Delta   \Theta_i  \rho\Bigg) \nonumber \\
		\le&\sup_\rho \sum_i q_i S\Big( \Theta_i    \Delta \rho\Big\| \Delta   \Theta_i  \rho\Big) \nonumber \\
		\le& \sum_i q_i \sup_\rho S\Big( \Theta_i    \Delta \rho\Big\| \Delta   \Theta_i  \rho\Big) \nonumber \\
		=&\sum_i q_i M_c(\Theta_i).
	\end{align}
\end{proof}

\section{On the need of resource theories on the level of operations}
In this section, we will underpin our claim in the main text that resource theories on the level of operations are
not only meaningful but essential unifying concepts in the theory of resources. Sticking to the spirit of this work, we will do this at the example of coherence.

First, let us clarify why we cannot quantify the coherence of an operation by the coherence of its corresponding Choi state (see Eq.~(\ref{eqA:Choi})).
Conceptually, the question whether this is possible can only be  answered \textit{after} one clarified how to quantify the coherence of operations, which we did in the main text.
The Choi-Jamiolkowski isomorphism~\cite{CHOI1975285,PhysRevA.87.022310} defines a one-to-one mapping between linear maps from $\mathbbm{C}^{n,n}$ to $\mathbbm{C}^{m,m}$ and linear operators on $\mathbbm{C}^{nm}$. This isomorphism allows for many useful associations, including a correspondence between CPTP maps and positive semidefinite operators~\cite{CHOI1975285}. However, this does not imply that we can associate all properties of maps straightforwardly with the ``similar'' properties of their corresponding Choi states. From a mathematical point of view, this should be clear: A priori, an isomorphism is simply a one-to-one mapping, and therefore, there is not necessarily a straight forward correspondence between ``similar'' properties.

In particular, the coherence of a quantum operation, i.e. its ability to create or detect coherence, cannot be identified in a trivial way with the coherence of its Choi state (although this is claimed in \cite{datta2017coherence}). Consider for example the identity channel, which can neither detect nor create coherence. The corresponding Choi state has coherence. However, the total dephasing operation $\Delta$, which is also incapable of detecting or creating coherence, corresponds to an incoherent Choi state. Of course it is possible to map conditions on a quantum operation to restrictions on the corresponding Choi state, which we exploited in the semidefinite program calculating the diamond-measure. Directly applying the isomorphism, it is easy to show that a Choi  state $J$ represents an detection-incoherent operation iff
\begin{align}
	J\ge 0,\qquad \tr_B J=\id_A, \qquad (\Delta_B\otimes\id_A) J=\Delta_{BA}J,
\end{align}
where the first condition ensures complete positivity, the second trace preservation and the third one detection-incoherence. 
There also exist operational settings in which the value of a quantum operation is described by the coherence of the corresponding Choi state, see for example~\cite{bera2018quantifying}. However, these settings are not related to the detection or creation of coherence.

Next we will demonstrate that it is impossible to quantify the coherence detection ability of operations in the known coherence theories that are formulated on the level of states. As discussed in the main text, this serves as an example that there exist resources which cannot be captured on the level of states.
To do this, we first need to clarify when an operation can detect coherence and when not. As demonstrated in the main text, a quantum operation can detect coherence iff it can transform coherences to populations. The coherence generation capacity of a quantum operation mentioned in the introduction of the main text is, as the name says, connected to the generation of coherence and therefore clearly not the figure of merit when we intend to analyze the detection of coherence. Indeed, there are operations that can create coherence but not detect it and vice versa.
The other quantity used to analyze the value of operations is the resource cost. However, as we will show in the following, this quantity is trivial for all sets of incoherent operations appearing in the literature, which we will list in the following.

The largest set of operations that do not generate coherence on average are called maximally incoherent operations (MIO)~\cite{aberg2006quantifying}. A subset of these operations are the incoherent operations (IO), which do not generate coherence even if subselection is allowed~\cite{BaumgratzPhysRevLett.113.140401}. The set of incoherent operations in turn contains the incoherent operations on the wire (LOP), which were introduced in the operationally motivated framework of local operations and physical wires~\cite{egloff2018local}.

Operations that cannot detect coherence on average (see this work and~\cite{PhysRevLett.118.060502}) contain as subsets DIO~\cite{ChitambarPhysRevLett.117.030401,PhysRevA.94.052324,PhysRevA.94.052336,PhysRevA.95.019902,PhysRevLett.110.070502,PhysRevLett.118.060502} introduced in the main text, strictly incoherent operations (SIO)~\cite{WinterPhysRevLett.116.120404}, translationally-invariant operations (TIO)~\cite{gour2008referenceFrames,PhysRevA.94.052324}, physical incoherent operations (PIO)~\cite{ChitambarPhysRevLett.117.030401}, genuinely incoherent operations (GIO)~\cite{1751-8121-50-4-045301}, fully incoherent operations (FIO)~\cite{1751-8121-50-4-045301} and energy preserving operations (EPO)~\cite{PhysRevA.96.022327}. See also the review article~\cite{StreltsovRevModPhys.89.041003} for an overview over these sets of operations and further references.

Let us begin with the resource cost of free operations that cannot detect coherence. 	
Assume that an operation $\Theta$ can be simulated by a fixed operation $\Phi\in \mathcal{DI}$ and the consumption of a fixed state $\rho$ in the sense that
\begin{align}
	\Theta(\sigma)=\Phi(\sigma\otimes\rho) \forall \sigma.
\end{align}
From this follows
\begin{align}\label{eq:NoSimulation}
	\Delta \Theta \sigma =\Delta \Phi (\sigma\otimes\rho)=\Delta\Phi\Delta(\sigma\otimes\rho)=\Delta\Phi\Delta(\Delta\sigma\otimes\rho)=\Delta\Phi(\Delta\sigma\otimes\rho)=\Delta \Theta \Delta \sigma,
\end{align}
where we made use of the fact that $\Phi \in \mathcal{DI}$ means $\Delta \Phi \Delta=\Delta \Phi$.
Since this has to be valid for all $\sigma$ (or equivalently for an unknown $\sigma$) by assumption, we have $\Theta \in \mathcal{DI}$, independent of the choice of $\rho$. Therefore we cannot simulate any operation that can detect coherence and the resource cost for these operations is not defined/infinity, turning the resource cost into a trivial measure. Since FIO, GIO, PIO, SIO and DIO as well as TIO and EFO are subsets of $\mathcal{DI}$, this is also the case for these sets of free operations.

On the contrary, the resource cost is strictly zero for MIO, IO and LOP: Since the ability of an operation to detect coherences depends only on the populations after the operation has been applied, $\Delta\Theta$ can measure coherences equally well as $\Theta$. Now take an arbitrary Kraus decomposition $K_n$ of $\Theta$. Then it is easy to check that $L_{i,n}=\ketbra{i}{i}K_n$ forms a set of incoherent Kraus operators implementing $\Delta\Theta$. Therefore measuring coherence is free in IO and consequently also in MIO. From the elemental free operations in LOP follows directly that measuring coherences on the wire is free as well and hence the resource cost is zero for LOP, IO and MIO as promised.
This means that the resource cost to measure coherence is the same for all states in all settings on the level of states, and therefore not a (faithful) measure.

The above arguments also show that it is impossible to construct any coherence theory on the level of states which leads to a resource cost well suited to quantify the coherence detection ability of operations: If any operation that can detect coherence is free, the resource cost is not faithful. Otherwise Eq.~(\ref{eq:NoSimulation}) applies and the resource cost is not defined.

In case we investigate the resource creation ability of operations, resource generation capacities become meaningful quantities and allow to quantify the value of operations. But even in this case, a direct quantification using a resource theory of operations has its advantage: Since there exists an infinitude of measures quantifying the amount of resources in states~\cite{vidal2000entanglementMonotones,du2015coherence,du2015erratum}, there exist also an infinitude of different resource generation capacities. This introduces a level of uncertainty about the quantification of coherence on the level of states which we can avoid when using a resource theory on the level of operations.

Next we want to clarify how our work relates to the concept of coherence witnesses introduced in \cite{PhysRevLett.116.150502}. As the name suggest, a coherence witness can be used to reveal the presence of coherences in certain states, but fails to detect it in others. Our framework could therefore be used to quantify how well a given coherence witness is able to detect coherences. 

Finally, let us show the relevance of coherence detecting operations in multi-path interferometric experiments, which are of considerable importance both in fundamental science and technology~\cite{michelson1887relative,PhysRevLett.59.2044,hariharan2010basics,PhysRevLett.116.061102adapted}. We identify the different paths which the particles can take inside the interferometer with pure incoherent states.
Following~\cite{Biswas20170170}, a (generalized) interferometric experiment consists of three steps: first a state with some coherence is created (for example by a  beam splitter), then some path-dependent phases are imprinted on the state and finally a measurement extracting information about these phases takes place (usually involving a second beam splitter and detectors).
This setup is an particularly easy example to demonstrate why both the ability to create coherence and the ability to detect it are valuable resources: If we cannot create coherence in the first step, the second step cannot imprint the information we are interested in onto the state, since path-dependent phases only affect coherences. In addition, if we cannot detect these coherences in the third step, we cannot acquire any information about the imprinted phases. It was analyzed in~\cite{Biswas20170170} how the coherence after the first step influences the obtainable information about the phases when optimal detection procedures are used in the third step. We leave it to future work to quantitatively connect the amount of information we can extract and the dynamic properties of coherence, including its detection. The goal of this is not to describe a new experiment, but to understand the relevance of quantum mechanical resources in technological applications. Such an improved understanding facilitates an improved design of new technologies. These will be far more complex than interferometric experiments and are assumed to give considerable advantages in communication, sensing, and computation. Our contribution consists in offering tools to analyze the potential of such technologies in a systematic way.

\end{document}